\documentclass{article}

\usepackage[preprint]{corl_2020} %

\usepackage{times}
\usepackage{setspace}
\usepackage{url}
\spacing{1}
\usepackage[utf8]{inputenc}
\usepackage[T1]{fontenc}
\usepackage{graphicx}		%
\usepackage{wrapfig}
\usepackage[format=plain,font=footnotesize,labelfont=bf,labelsep=period]{caption}
\usepackage{sidecap} %
\usepackage{subfig}
\usepackage[export]{adjustbox}
\usepackage[font=small]{caption}
\usepackage{float}

\usepackage{amsmath} %
\usepackage{amssymb}  %
\usepackage{amsthm}
\usepackage{mathtools}
\usepackage[normalem]{ulem}
\usepackage{paralist}	%
\usepackage[space]{grffile} %
\usepackage{color}

\newtheorem{theorem}{Theorem}
\newtheorem*{theorem*}{Theorem}
\newtheorem{corollary}{Corollary}
\newtheorem{lemma}{Lemma}
\theoremstyle{definition}
\newtheorem{definition}{Definition}
\theoremstyle{remark}

\theoremstyle{definition}

\theoremstyle{definition}
\newtheorem{proposition}{Proposition}

\newcommand{\R}{\mathbb{R}}
\newcommand{\C}{\mathcal{C}}

\newcommand{\K}{\mathcal{K}}

\definecolor{blue}{RGB}{38,38,134}
\definecolor{darkblue}{RGB}{0,0,102}
\definecolor{lightblue}{RGB}{77,77,148}

\definecolor{gold}{RGB}{234, 170, 0}
\definecolor{metallic_gold}{RGB}{139, 111, 78}

\renewcommand{\cal}[1]{\mathcal{ #1 }}
\newcommand{\mb}[1]{\mathbf{ #1 }}
\newcommand{\bs}[1]{\boldsymbol{ #1 }}

\newcommand{\derp}[2]{\frac{\partial #1 }{\partial #2 }}

\newcommand{\simiid} {\stackrel{\textrm{iid}}{\sim}}

\DeclareMathOperator*{\argmin}{argmin}

\allowdisplaybreaks

\newcommand{\lmat}{\begin{bmatrix}}
\newcommand{\rmat}{\end{bmatrix}} 
\title{Guaranteeing Safety of Learned Perception Modules  via Measurement-Robust Control Barrier Functions}

\author{
  $^1$Sarah Dean~~~~$^2$Andrew J. Taylor~~~~$^2$Ryan K. Cosner \\ \\
  \bf{$^1$Benjamin Recht~~~~$^2$Aaron D. Ames} \\ \\
  $^1$University of California, Berkeley~~~~$^2$California Institute of Technology \\
  \texttt{\{dean\_sarah,brecht\}@berkeley.edu} \\  \texttt{\{ajtaylor,rkcosner,ames\}@caltech.edu}
}

\begin{document}
\maketitle

\begin{abstract}
Modern nonlinear control theory seeks to develop feedback controllers that endow systems with properties such as safety and stability. The guarantees ensured by these controllers often rely on accurate estimates of the system state for determining control actions. In practice, measurement model uncertainty can lead to error in state estimates that degrades these guarantees. In this paper, we seek to unify techniques from control theory and machine learning to synthesize controllers that achieve safety in the presence of measurement model uncertainty. We define the notion of a Measurement-Robust Control Barrier Function (MR-CBF) as a tool for determining safe control inputs when facing measurement model uncertainty. Furthermore, MR-CBFs are used to inform sampling methodologies for learning-based perception systems and quantify tolerable error in the resulting learned models. We demonstrate the efficacy of MR-CBFs in achieving safety with measurement model uncertainty on a simulated Segway system.
\end{abstract}


\section{Introduction}

Ensuring safety is of utmost importance in modern control systems, and as system complexity increases, it is necessary to rigorously encode safety during the controller design process. Examples of safety-critical control applications include autonomous vehicles, aerospace vehicles, and industrial robotics. In practice, these control systems rely on feedback involving imperfect or uncertain measurements models, which can lead to inaccurate state estimation and unsafe 
behavior if not properly accounted for in controller design. Furthermore, modern control challenges increasingly call for complex measurement systems incorporating data-driven learning methods, such as perception. Thus it is paramount that controllers and learning models be robustly designed to ensure safety in the presence of uncertain measurement models.

Control Barrier Functions (CBFs) \cite{ames2014control,ames2017control} have become increasingly popular \cite{ames2019control,singletary2020control,nguyen2016exponential} as a tool for achieving safety in the form of \textit{set invariance} \cite{blanchini1999set}. Furthermore, the integration of CBFs with machine-learning approaches for reducing model uncertainty has shown great promise theoretically \cite{taylor2020control, cheng2019accelerating, cheng2019end} and in application \cite{taylor2020learning, choi2020reinforcement}. Synthesis of CBFs that encode certain safety properties have been explored from the perspective of backup sets \cite{gurriet2018online} and data-driven methods \cite{robey2020learning}. In many of these settings it is assumed that controllers synthesized via CBFs have perfect access to measurements of the system state. In practice, these measurements are often corrupted due to uncertainty in measurement models and sensor noise. Safety guarantees in the presence of measurement noise has been addressed from a stochastic perspective \cite{clark2019control, nilsson2020lyapunov}, and robustness to error in the estimate of the state was considered via a sub-tangentiality condition in \cite{gurriet2019realizable}. The work in \cite{jankovic2018robust,takano2018robust} considers robust CBF formulations where worst-case disturbance bounds and perfect knowledge of the actuated dynamics enable robust safety. This existing work has not considered the case when the actuated dynamics are uncertain, nor a data-driven approach where deterministic error in the state estimate can be mitigated through learning.

In this work, we consider safety for a setting in which the system states are not directly observed.
Inspired by examples of vision-based control \cite{codevilla2018end,lambert2018deep,tang2018aggressive}, we suppose that state information is observed through a complex transformation, e.g. a camera image, and that an inverse mapping from measurement to state must be estimated.
While many impressive demonstrations in the context of robotics rely on carefully calibrated systems in such settings~\cite{roth1987overview},
rigorous investigations relating data-driven calibration to learning in the context of control are just beginning to receive attention.
This setting has been studied from the perspective of robustness and sample complexity for linear systems, where the main safety concern is stability~\citep{dean2019robust,dean2020robust}.
Other work has considered methods for ensuring obstacle avoidance under loss of observability for nonlinear systems \cite{laine2020eyes}, but do not consider learning directly. To the best of our knowledge, no such analysis exists for nonlinear systems through the perspective of CBFs.

The main contributions of this work are threefold. Firstly, we present the novel definition of Measurement-Robust Control Barrier Functions (MR-CBF) which modify the definition of CBFs to enable robustness to error in measurement models. Secondly, we show how MR-CBFs can be incorporated into convex optimization-based controllers that can be efficiently solved online. Lastly, we outline how MR-CBFs can be used to guide data sampling methodology in order to ensure that the resultant model yields tolerable error. 

This remainder of this paper is organized as follows. In Section \ref{sec:background} we provide a review of Control Barrier Functions and their use in synthesizing safe controllers. Section \ref{sec:MR-CBF} explores the impact of measurement model uncertainty on safety, and defines the MR-CBF as a tool for synthesizing controllers that guarantee safety in the absence of perfect state observation. In Section \ref{sec:learning} we explain how data-driven learning methods can be used to reduce measurement model uncertainty and yield valid MR-CBFs. Finally, Section \ref{sec:result} presents simulation results which demonstrate the efficacy of MR-CBFs in enforcing safety in the presence of measurement model uncertainty for a Segway system. Additional details including proofs are provided in supplementary appendices. 

\section{Background} 
\label{sec:background}
In this section we provide a review of safety and Control Barrier Functions (CBFs). These definitions will be used in quantifying how measurement uncertainty impacts safety guarantees. 

Consider the nonlinear control affine system given by:
\begin{equation}
    \label{eqn:eom}
    \dot{\mb{x}} = \mb{f}(\mb{x})+\mb{g}(\mb{x})\mb{u},
\end{equation}
where $\mb{x}\in\R^n$, $\mb{u}\in\R^m$, and $\mb{f}:\R^n\to\R^n$ and $\mb{g}:\R^n\to\R^{n\times m}$ are locally Lipschitz continuous on $\R^n$. Given a locally Lipschitz continuous state-feedback controller $\mb{k}:\R^n\to\R^m$, the closed-loop system dynamics are:
\begin{equation}
    \label{eqn:cloop}
    \dot{\mb{x}} = \mb{f}_{\textrm{cl}}(\mb{x}) \triangleq  \mb{f}(\mb{x})+\mb{g}(\mb{x})\mb{k}(\mb{x}).
\end{equation}
The assumption on local Lipschitz continuity of $\mb{f}$, $\mb{g}$, and $\mb{k}$ implies that $\mb{f}_\textrm{cl}$ is locally Lipschitz continuous. Thus for any initial condition $\mb{x}_0 \triangleq \mb{x}(0) \in \R^n$ there exists a time interval $I(\mb{x}_0) = [0, t_{\textrm{max}})$ such that $\mb{x}(t)$ is the unique solution to \eqref{eqn:cloop} on $I(\mb{x}_0)$ \cite{perko2013differential}.

The notion of safety that we consider is formalized by specifying a \textit{safe set} in the state space that the system must remain in to be considered safe. In particular, consider a set $\C\subset \R^n$ defined as the 0-superlevel set of a continuously differentiable function $h:\R^n \to \R$, yielding:
\begin{equation}
    \C \triangleq \left\{\mb{x} \in \R^n : h(\mb{x}) \geq 0\right\}. \label{eqn:safeset}
\end{equation}
We refer to $\C$ as the \textit{safe set}, and note $ \partial\C \triangleq \{\mb{x} \in \R^n : h(\mb{x}) = 0\}$ and $ \textrm{Int}(\C) \triangleq \{\mb{x} \in \R^n : h(\mb{x}) > 0\}$. We assume that $\C$ is nonempty and has no isolated points, that is, $\textrm{Int}(\C) \not = \emptyset \textrm{ and }\overline{\textrm{Int}(\C)} = \C$. This construction motivates the following definitions of forward invariance and safety:

\begin{definition}[\textit{Forward Invariant \& Safety}]
A set $\C\subset\R^n$ is \textit{forward invariant} if for every $\mb{x}_0\in\C$, the solution $\mb{x}(t)$ to \eqref{eqn:cloop} satisfies $\mb{x}(t) \in \C$ for all $t \in I(\mb{x}_0)$. The system \eqref{eqn:cloop} is \textit{safe} with respect to the set $\C$ if the set $\C$ is forward invariant.
\end{definition}

Before defining Control Barrier Functions as a tool for synthesizing safe controllers, we note that a continuous function $\alpha:(-b,a)\to\R$, with $a,b>0$, is said to belong to \textit{extended class $\cal{K}$} ($\alpha\in\cal{K}_e$) if $\alpha(0)=0$ and $\alpha$ is strictly monotonically increasing. If $a,b=\infty$, $\lim_{r\to\infty}\alpha(r)=\infty$, and $\lim_{r\to-\infty}\alpha(r)=-\infty$, then $\alpha$ is said to belong to \textit{extended class $\cal{K}_\infty$} ($\alpha\in\cal{K}_{\infty,e}$). Furthermore, we note $c\in\R$ is referred to as a \textit{regular value} of a continuously differentiable function $h:\R^n\to\R$ if $h(\mb{x})=c\implies\derp{h}{\mb{x}}(\mb{x})\neq\mb{0}$. These enable the definition of Control Barrier Functions as follows:

\begin{definition}[\textit{Control Barrier Function (CBF)}, \cite{ames2017control}]\label{def:cbf}
Let $\C\subset\R^n$ be the 0-superlevel set of a continuously differentiable function $h:\R^n\to\R$ with $0$ a regular value. The function $h$ is a \textit{Control Barrier Function} (CBF) for \eqref{eqn:eom} on $\C$ if there exists $\alpha\in\K_{\infty,e}$ such that for all $\mb{x}\in\C$:
\begin{equation}
\label{eqn:cbf}
     \sup_{\mb{u}\in\R^m} \dot{h}(\mb{x},\mb{u}) \triangleq \underbrace{\derp{h}{\mb{x}}(\mb{x})\mb{f}(\mb{x})}_{L_\mb{f}h(\mb{x})}+\underbrace{\derp{h}{\mb{x}}(\mb{x})\mb{g}(\mb{x})}_{L_\mb{g}h(\mb{x})}\mb{u}>-\alpha(h(\mb{x})).
\end{equation}
\end{definition}
This condition can be equivalently stated as:
\begin{equation}
\label{eqn:cbfalt}
    \Vert L_\mb{g}h(\mb{x}) \Vert_2 = 0 \implies L_\mb{f}h(\mb{x}) > -\alpha(h(\mb{x})). 
\end{equation}
The inequality in the definition of CBFs is strict to ensure that the controllers synthesized via CBFs are locally Lipschitz continuous \cite{jankovic2018robust}. Given a CBF $h$ for \eqref{eqn:eom} and a corresponding $\alpha\in\cal{K}_{\infty,e}$, we can consider the pointwise set of all control values that satisfy \eqref{eqn:cbf}:
\begin{equation}
    K_{\textrm{cbf}}(\mb{x}) \triangleq \left\{\mb{u}\in\R^m ~\left|~ L_\mb{f}h(\mb{x})+L_\mb{g}h(\mb{x})\mb{u}\geq-\alpha(h(\mb{x})) \right. \right\}.
\end{equation}
A main result in \cite{xu2015robustness, ames2014control} relates controllers taking values in  $K_{\textrm{cbf}}(\mb{x})$ to the safety of \eqref{eqn:cloop} on $\C$:
\begin{theorem}\label{thm:cbf_safe}
Given a set $\C\subset\R^n$ defined as the 0-superlevel set of a continuously differentiable function $h:\R^n\to\R$, if $h$ is a CBF for \eqref{eqn:eom} on $\C$, then any locally Lipschitz continuous controller $\mb{k}:\R^n\to\R^m$, such that $\mb{k}(\mb{x})\in K_{\textrm{cbf}}(\mb{x})$ for all $\mb{x}\in\C$, renders the system \eqref{eqn:cloop} safe w.r.t. $\C$.
\end{theorem}

This result motivates the construction of a pointwise optimal controller seeking to minimize a cost associated with the choice of input. To this end, we consider the \textit{safety-critical control} formulation \cite{gurriet2018towards} that seeks to filter a hand-designed but potentially unsafe locally Lipschitz continuous controller, $\mb{k}_d:\R^n\to\R^m$, to find the nearest safe action:
\begin{align}
\label{eqn:CBF-QP}
\tag{CBF-QP}
\mb{k}(\mb{x}) =  \,\,\underset{\mb{u} \in \R^m}{\argmin}  &  \quad \frac{1}{2} \Vert \mb{u}-\mb{k}_d(\mb{x}) \Vert_2^2  \\
\mathrm{s.t.} \quad & \quad L_\mb{f}h(\mb{x})+L_\mb{g}h(\mb{x})\mb{u}\geq-\alpha(h(\mb{x})). \nonumber
\end{align}
The validity of $h$ as a CBF ensures the feasibility of this optimization problem, and the resulting controller is locally Lipschitz continuous \cite{jankovic2018robust}.
    
\section{Measurement-Robust Control Barrier Functions}
\label{sec:MR-CBF}
In this section we explore the impact of measurement model uncertainty on safety guarantees, and propose the notion of a modified Control Barrier Function that is robust to such errors.

In many practical applications, the state $\mb{x}$ is not directly available to the controller, but rather a state-dependent sensor measurement: 
\begin{equation}
    \mb{y} = \mb{p}(\mb{x}),
\end{equation}
where $\mb{p}:\R^n\to\R^k$ is assumed to be locally Lipschitz continuous. We assume the relationship between the measurement and the true state is deterministic, and note the application of CBFs in the context of stochastic differential equations has been considered in \cite{clark2019control}. Future work will consider the unification of the results in this paper with the stochastic setting. We further assume there exists a locally Lipschitz continuous function $\mb{q}:\R^k\to\R^n$ such that for all $\mb{x}\in\R^n$, we have $\mb{q}(\mb{p}(\mb{x}))=\mb{x}$. This assumption implies that the state can be uniquely determined from any given measurement. This bijective relationship would allow the measurements to be redefined as the state of the system if the function $\mb{p}$ was known, but that is often not the case in many modern control applications (such as when using vision).

While the function $\mb{p}$ is often determined by the physical attributes of a system, a locally Lipschitz continuous estimate of the function $\mb{q}$, given by $\widehat{\mb{q}}:\R^k\to\R^n$, is constructed to determine an estimate of the state, $\widehat{\mb{x}}$. For notational simplicity we define the measurement-estimate function $\widehat{\mb{v}}:\R^n\to\R^k\times\R^n$ such that $\widehat{\mb{v}}(\mb{x})=(\mb{p}(\mb{x}),\widehat{\mb{q}}(\mb{p}(\mb{x}))$. We also define the set $\mb{p}(\mathcal{C})\subset\R^k$ as the image of the safe set under the measurement function, the set $\widehat{\mb{q}}(\mb{p}(\mathcal{C}))\subset\R^n$ as the image of the safe set under the state estimate function, and $\widehat{\mb{v}}(\C)$ as the image of the safe set under the measurement-estimate function.

The function $\widehat{\mb{q}}$ is constructed either via system and measurement models, or from data using learning methods, and thus its accuracy in estimating $\mb{q}$ degrades with imperfections in sensor fabrication and integration, or imperfections in learning models and training data. Thus we assume that our state estimate is related to the true state as follows:
\begin{equation}
    \widehat{\mb{x}} \triangleq \widehat{\mb{q}}(\mb{y}) = \mb{x}+\mb{e}(\mb{x}),
\end{equation}
for an unknown function $\mb{e}:\R^n\to\R^n$ that is defined implicitly via $\widehat{\mb{q}}$. In practice, the function $\mb{e}$ can often be characterized via upper bounds on model uncertainty or via data-driven arguments for learning models (discussed further in Section \ref{sec:learning}).  In particular, we assume that while $\mb{e}(\mb{x})$ is not known for a particular value of $\mb{x}$, it is known that $\mb{e}(\mb{x})\in\mathcal{E}(\mb{y})$ for a measurement dependent, compact pointwise set $\mathcal{E}(\mb{y})$.
This leads to the definition of the following two pointwise sets:
\begin{align}
    \widehat{\mathcal{X}}(\mb{x}) \triangleq &  \left\{\widehat{\mb{x}}\in\R^n~|~\exists~\mb{e}\in\mathcal{E}(\mb p(\mb x)) ~\textrm{s.t.}~ \widehat{\mb{x}}=\mb{x}+\mb{e} \right\}, \label{eqn:Xhat}\\ \mathcal{X}(\mb{y}) \triangleq &  \left\{\mb{x}\in\R^n~|~\exists~\mb{e}\in\mathcal{E}(\mb{y}) ~\textrm{s.t.}~ \widehat{\mb{x}}=\mb{x}+\mb{e} \right\}. \label{eqn:X}
\end{align}
The first of these two pointwise sets can be interpreted as all possible state estimates corresponding to a particular state, restricted by the possible error dictated by $\mathcal{E}(\mb p(\mb x))$. 
While it is not directly computable without knowledge of $\mb p$, this set will play an important conceptual role in Section \ref{sec:learning} to argue about how data can be used to determine error bounds. 
The second pointwise set consists of all potential states that may yield a measurement-state estimate pair. 

Since a controller enforcing the CBF condition \eqref{eqn:cbf} requires exact knowledge of the state $\mb{x}$, we propose an alternative condition which depends on only the set $\mathcal{X}(\mb{y})$ and the state estimate $\widehat{\mb{x}}$.
To ensure safety with a CBF, it is sufficient for the following condition to hold for all $\mb{y} \in\mb{p}(\C)$:
\begin{equation} 
    \sup_{\mb{u}\in\R^m} \inf_{\mb{x}\in\mathcal{X}(\mb{y})} \derp{h}{\mb{x}}(\mb{x})(\mb{f}(\mb{x}) + \mb{g}(\mb{x}) \mb{u}) +\alpha(h(\mb{x})) \geq 0\:.
\end{equation}
This condition implies that there exists a control input that renders the system safe for all possible states corresponding to a given state estimate. Verifying that this condition holds can be difficult for an arbitrary CBF, and it is not easily (or possibly) enforced in a convex-optimization based controller. To resolve these problems, we introduce the following definition: 

\begin{definition}[\textit{Measurement-Robust Control Barrier Function (MR-CBF)}]
Let $\C\subset\R^n$ be the 0-superlevel set of a continuously differentiable function $h:\R^n\to\R$ with $0$ a regular value. The function $h$ is a \textit{Measurement-Robust Control Barrier Function} (MR-CBF) for \eqref{eqn:eom} on $\C$ with \textit{parameter function} $(a,b):\R^k\to\R^2_+$ if there exists $\alpha\in\K_{\infty,e}$ such that for all $(\mb{y},\widehat{\mb{x}})\in\widehat{\mb{v}}(\C)$:
\begin{equation}
\label{eqn:mr-cbf}
     \sup_{\mb{u}\in\R^m} L_\mb{f}h(\widehat{\mb{x}})+L_\mb{g}h(\widehat{\mb{x}})\mb{u} - (a(\mb{y}) + b(\mb{y}) \Vert\mb{u}\Vert_2)  >-\alpha(h(\widehat{\mb{x}})).
\end{equation}
\end{definition}
Verifying this condition over $\widehat{\mb{v}}(\C)$ may not be possible, but as will be seen in Section \ref{sec:learning},  the set $\widehat{\mathcal{X}}(\mb x)$ can be used to provide sufficient conditions under which~\eqref{eqn:mr-cbf} is met. The definition of a MR-CBF introduces the non-positive term $-(a(\mb{y})+b(\mb{y})\Vert\mb{u}\Vert_2)$ to the control barrier condition, requiring that a stronger degree of safety be enforced compared to the typical control barrier condition. Furthermore, the norm of the input appears in this term, indicating that for large values of $b$ large inputs can lead to unsafe behavior. This condition is equivalently stated as:
\begin{equation}
\label{eqn:mrcbfalt}
    \Vert L_\mb{g}h(\widehat{\mb{x}}) \Vert_2 \leq b(\mb{y}) \implies L_\mb{f}h(\widehat{\mb{x}}) > -\alpha(h(\widehat{\mb{x}})) + a(\mb{y}). 
\end{equation}
In contrast to the implication in \eqref{eqn:cbfalt}, the size of set for which the antecedent in \eqref{eqn:mrcbfalt} is met may be larger, requiring the natural dynamics to be safe ($L_\mb{f}h(\widehat{\mb{x}}) > -\alpha(h(\widehat{\mb{x}})) + a(\mb{y})$) in a larger region. Given a MR-CBF $h$ for \eqref{eqn:eom} on $\C$ with parameter function $(a,b)$ and a corresponding $\alpha\in\cal{K}_{\infty,e}$, we can consider the pointwise set of all control values that satisfy \eqref{eqn:mr-cbf}:
\begin{equation}
\label{eqn:Kmrcbf}
    K_{\textrm{mr-cbf}}(\mb{y},\widehat{\mb{x}}) \triangleq \left\{\mb{u}\in\R^m ~\left|~ L_\mb{f}h(\widehat{\mb{x}})+L_\mb{g}h(\widehat{\mb{x}})\mb{u}- ( a(\mb{y}) + b(\mb{y}) \Vert\mb{u}\Vert_2)\geq-\alpha(h(\widehat{\mb{x}})) \right. \right\},
\end{equation}
for $(\mb{y},\widehat{\mb{x}})\in\widehat{\mb{v}}(\C)$. Given this construction, we have the following result relating the existence of a MR-CBF to safety under the presence of measurement model uncertainty:

\begin{theorem}\label{thm:safety}
Let a set $\C\subset\R^n$ be defined as the 0-superlevel set of a continuously differentiable function $h:\R^n\to\R$. Assume the functions $L_\mb{f}h:\R^n\to\R$, $L_\mb{g}h:\R^n\to\R^m$, and $\alpha\circ h:\R^n\to\R$ are Lipschitz continuous on $\C$ with Lipschitz coefficients $\mathfrak{L}_{L_\mb{f}h}$, $\mathfrak{L}_{L_\mb{g}h}$, and $\mathfrak{L}_{\alpha\circ h}$, respectively. Further assume there exists a locally Lipschitz function $\epsilon:\R^k\to\R_+$, such that $\max_{\mb{e}\in\mathcal{E}(\mb{y})}\Vert\mb{e}\Vert_2\leq\epsilon(\mb{y})$ for all $\mb{y} \in \mb{p}(\C)$. If $h$ is a MR-CBF for \eqref{eqn:eom} on $\C$ with parameter function $(\epsilon(\mb{y})(\mathfrak{L}_{L_\mb{f}h}+\mathfrak{L}_{\alpha\circ h}),\epsilon(\mb{y})\mathfrak{L}_{L_\mb{g}h})$, then any locally Lipschitz continuous controller $\mb{k}:\R^k\times\R^n\to\R^m$, such that $\mb{k}(\mb{y},\widehat{\mb{x}})\in K_{\textrm{mr-cbf}}(\mb{y},\widehat{\mb{x}})$ for all $(\mb{y},\widehat{\mb{x}})\in\widehat{\mb{v}}(\C)$, renders the system \eqref{eqn:cloop} safe w.r.t. $\C$.
\end{theorem}

A proof of this theorem can be found in Appendix \ref{app:mrcbfproof}. To more clearly see how the upper bound on the estimate error, $\epsilon(\mb{y})$, manifests in the MR-CBF condition, we note the particular condition that must be satisfied for this theorem is given by:
\begin{equation}
    \label{eqn:mrcbflip}
     \sup_{\mb{u}\in\R^m} L_\mb{f}h(\widehat{\mb{x}})+L_\mb{g}h(\widehat{\mb{x}})\mb{u}-\epsilon(\mb{y})(\mathfrak{L}_{L_\mb{f}h}+\mathfrak{L}_{\alpha\circ h}+\mathfrak{L}_{L_\mb{g}h}\Vert\mb{u}\Vert_2)>-\alpha(h(\widehat{\mb{x}})).
\end{equation}
Thus as $\epsilon(\mb{y})$ becomes smaller, the level of robustness required by an MR-CBF approaches that of a regular CBF for the same set $\C$, and recovers the original control barrier condition with no estimate error. Furthermore, smaller values of $\epsilon(\mb{y})$ can be interpreted as leading to an enlarging of the region over which the condition \eqref{eqn:mrcbfalt} holds.

One advantage of this approach for resolving the impact of measurement model uncertainty on safety is that the constraint in \eqref{eqn:Kmrcbf} remains convex. This constraint can then be directly integrated into an optimization based controller as follows:
\begin{align}
\label{eqn:MR-OP}
\tag{MR-OP}
\mb{k}(\mb{y},\widehat{\mb{x}}) =  \,\,\underset{\mb{u} \in \R^m}{\argmin}  &  \quad \frac{1}{2} \Vert \mb{u}-\mb{k}_d(\widehat{\mb{x}}) \Vert_2^2  \\
\mathrm{s.t.} \quad & \quad L_\mb{f}h(\widehat{\mb{x}})-(\mathfrak{L}_{L_\mb{f}h}+\mathfrak{L}_{\alpha\circ h})\epsilon(\mb{y})+L_\mb{g}h(\widehat{\mb{x}})\mb{u}-\mathfrak{L}_{L_\mb{g}h}\epsilon(\mb{y})\Vert\mb{u}\Vert_2\geq-\alpha(h(\widehat{\mb{x}})). \nonumber
\end{align}
This problem is in fact a second-order cone program (SOCP), with an explicit conversion to standard form provided in Appendix \ref{app:mropsocp}. As this constraint is non-smooth, existing methods for computing closed-form solutions via Lagrangian duality and assessing Lipschitz continuity \cite{jankovic2018robust} are not applicable. Future work will consider methods from variational analysis to study the Lipschitz continuity of solutions to this problem \cite{mordukhovich2014full}. In practice, a slack variable, $\delta$, is often added to ensure constraint feasibility. This relaxation is penalized in the cost with a large coefficient $p\in\R_{++}$:
\begin{align}
\label{eqn:R-MR-OP}
\tag{R-MR-OP}
\mb{k}(\mb{y},\widehat{\mb{x}}) =  \,\,\underset{(\mb{u},\delta) \in \R^m\times\R}{\argmin}  &  \quad \frac{1}{2} \Vert \mb{u}-\mb{k}_d(\widehat{\mb{x}}) \Vert_2^2 + p\delta^2  \\
\mathrm{s.t.} \quad   L_\mb{f}h&(\widehat{\mb{x}})-(\mathfrak{L}_{L_\mb{f}h}+\mathfrak{L}_{\alpha\circ h})\epsilon(\mb{y})+L_\mb{g}h(\widehat{\mb{x}})\mb{u}-\mathfrak{L}_{L_\mb{g}h}\epsilon(\mb{y})\Vert\mb{u}\Vert_2\geq-\alpha(h(\widehat{\mb{x}}))-\delta. \nonumber
\end{align}

While this relaxed controller does not necessarily enforce the desired safety constraint, if $\delta$ remains small the impact on safety can be understood through the notion of projection-to-state safety \cite{taylor2020control}. Furthermore, this relaxation ensures that the resulting controller is locally Lipschitz continuous as made explicit in Appendix \ref{app:rmroplipschitz} via the methods in \cite{dontchev1993lipschitzian}.
   
\section{Learning for Measurement Model Uncertainty Reduction}
\label{sec:learning}
In this section we explore how data-driven learning methods can be used to reduce measurement model uncertainty and yield valid MR-CBFs.

The previous section provides a method for guaranteeing safety in the absence of perfect state observation which relies on specifying a valid MR-CBF $h$.
The condition in~\eqref{eqn:mr-cbf} restricts admissible $h$ and $\C$ with a stronger condition than the typical CBF condition.
A CBF $h$ which is valid when states are perfectly observed may no longer be valid in the presence of measurement model uncertainty.
In this section, we reconsider the requirement that $h$ be a valid MR-CBF as a specification on measurement or calibration errors.
Clearly, as measurement model uncertainty becomes arbitrarily close to $0$, any valid CBF will be a valid MR-CBF.
We now make this intuition precise.

By applying the logic of the implication~\eqref{eqn:mrcbfalt}, we see that \eqref{eqn:mrcbflip} is true as long as for all  $(\mb{y},\widehat{\mb{x}})\in\widehat{\mb{v}}(\C)$,
\begin{equation}\label{eq:esp_y_bd}
      \epsilon(\mb y) < \max\left\{
      \frac{\Vert L_\mb{g}h(\widehat{\mb{ x}})\Vert_2}{\mathfrak{L}_{L_\mb{g}h}},
      \frac{L_\mb{f}h(\widehat{\mb{x}}) +  \alpha(h(\widehat{\mb{ x}}))}{\mathfrak{L}_{L_\mb{f}h} + \mathfrak{L}_{\alpha \circ h}}
      \right\}.
\end{equation}
This expression gives the maximum admissible error under which the MR-CBF condition will hold for a given set of dynamics and function $h$.
It is not easy to reason about this quantity directly, since the measurement-estimate set $\widehat{\mb{v}}(\C)$ is difficult to characterize without knowledge of $\mb{p}$.
Instead, we reformulate this expression to depend on the associated underlying state $\mb{x}$.
We do so by interpreting $\epsilon$ as a function of the state via $\epsilon(\mb{y}) = \epsilon(\mb{p}(\mb{x})) \triangleq \varepsilon(\mb{x})$ and appealing conceptually to the set $\widehat{\mathcal X}(\mb{x})$ defined in \eqref{eqn:Xhat} to consider all possible observations at a given state.

\begin{theorem}\label{thm:err_bd_x}
Let $\C$, $h$, $\mathfrak{L}_{L_\mb{f}h}$, $\mathfrak{L}_{L_\mb{g}h}$, $\mathfrak{L}_{\alpha\circ h}$, and $\epsilon$ be defined as in Theorem~\ref{thm:safety}.
Then $h$ is a MR-CBF for \eqref{eqn:eom} on $\C$ with parameter function $(\epsilon(\mb y)(\mathfrak{L}_{L_\mb{f}h}+\mathfrak{L}_{\alpha\circ h}),\epsilon(\mb{y})\mathfrak{L}_{L_\mb{g}h})$ if for all $\mb x \in \C$:
\begin{equation}\label{eq:upper_eps_x}
      \varepsilon(\mb{x}) < \max\left\{
      \frac{\Vert L_\mb{g}h({\mb{ x}})\Vert_2}{2\mathfrak{L}_{L_\mb{g}h}},
      \frac{L_\mb{f}h({\mb{x}}) +  \alpha(h({\mb{ x}}))}{2(\mathfrak{L}_{L_\mb{f}h} + \mathfrak{L}_{\alpha \circ h})}
      \right\} \triangleq \bar{\varepsilon}(\mb x)\:.
\end{equation}
\end{theorem}

Even though the true states will not be observed during operation, this result
gives a sufficient condition on error with respect to states in the safe set $\mathcal C$.
Comparing this expression with the bound in~\eqref{eq:esp_y_bd}, there is an additional factor of $2$.
This corresponds to doubling the radius of the uncertainty set, which occurs because it is necessary to consider the entire set $\widehat{\mathcal{X}}(\mb x)$ for each $\mb x$.

We now draw an explicit connection between measurement model uncertainty and learning from data.
Recall that the error $\mb e(\mb x)$ is implicitly defined by the estimated $\widehat{\mb q}$, so it can be viewed as arising from imperfectly approximating the inverse map ${\mb q}$:
\begin{equation}\label{eqn:pw_err_bound}
    \mb{e}(\mb{x}) = \mb{x} - \widehat{\mb{q}}(\mb{p}(\mb{x}))= \mb{q}(\mb{y}) - \widehat{\mb{q}}(\mb{y}).
\end{equation}
Therefore, we treat $\varepsilon(\mb{x})$ as a pointwise error bound on the learned map $\widehat{\mb{q}}$ over the space of measurements.
In a supervised learning setting, this map arises from training data, which we denote as $\mathcal S =\{(\mb{y}_i, \bar{\mb{x}}_i)\}_{i=1}^N$. The recorded system outputs $\mb{y}_i$ play the role of the independent variables, while the corresponding recorded states $\bar{\mb{x}}_i$ play the role of dependent variables.
We denote the recorded states as $\bar{\mb{x}}_i$ to allow for the possibility that they are imperfect measurements of the state.
For example, the recorded values could be noisy according to some distribution $\mathcal{D}$:
\begin{equation}\label{eqn:noisy_x}
    \bar{\mb{x}}_i = \mb{x}_i + \mb{w}_i,\quad \mb{w}_i \simiid \mathcal{D}\:,
\end{equation}
for $i=1,\dots,N$.
This scenario corresponds to using a simple noisy sensor to calibrate a more complex sensor by learning the inverse map $\widehat{\mb q}$.
Denote the 
true states as
$\mathcal X_{\mathcal S} = \{\mb{x}_i\}_{i=1}^N$. 

Our work is motivated by instances of control-from-pixels, where the system outputs $\mb y$ are camera images \cite{codevilla2018end,lambert2018deep,tang2018aggressive}.
Therefore, we consider maps $\widehat{\mb q}$ which are non-parametric and thus highly data dependent, like neural networks.\footnote{While we suppose that $\widehat{\mb{q}}$ is learned entirely from data, our results are equally applicable to learning the \textit{residual errors} of an existing perception component (see, e.g. surveyed methods in~\cite{de2019learning}).
}
Though necessary for guaranteeing safety,
providing pointwise bounds on errors as in \eqref{eqn:pw_err_bound} is often not the focus of machine learning analyses, which favor a mean-error perspective.
We point to this issue to highlight an area of future work, and rely here on a simplified model.
The following definition models training data-dependent errors:

\begin{definition}\label{def:nonparametric}
A \textit{non-parametric error bound} with parameters $(L , \sigma)\in\R^2_+$ and hyperparameter $\gamma\in\R_+$ has the form:
\begin{equation}
    \Vert\mb{e}(\mb{x})\Vert_2 \leq L\gamma +  \frac{\sigma}{\sqrt{\#\{\mb x_i \in \mathcal X_{\mathcal S} \mid \Vert\mb{x} - \mb{x}_i\Vert\leq \gamma\}}}\:.
\end{equation}
\end{definition}

In Appendix~\ref{app:feasibleproof}, we relate this definition to a high probability bound for a particular class of non-parametric regression models analysed in \cite{dean2020robust} for the noisy sensor setting described by \eqref{eqn:noisy_x}. 
Combining this definition with Theorem~\ref{thm:err_bd_x} yields a condition on the way that training data is collected.
\begin{corollary}\label{coro:dense_data}
Consider the setting of Theorem~\ref{thm:err_bd_x} and suppose that $\widehat{\mb{q}}$ is learned with a method that satisfies the $(L, \sigma)$ non-parametric error bound with hyperparameter $\gamma$.
Further suppose that the training data is collected from states $\mathcal{X}_{\mathcal{S}}$ such that for all $\mb{x}\in\C$,
\begin{align}\label{eq:density_bd}
    \#\{\mb{x}_i \in \mathcal{X}_{\mathcal{S}} \mid \Vert\mb{x} - \mb{x}_i\Vert_2\leq \gamma\} > \frac{4\sigma^2}{ \left(\max\Big\{
    \frac{1}{\mathfrak{L}_{L_\mb{g}h} }\left \Vert  L_\mb{g}h(\mb{x}) \right\Vert_2,
    \frac{L_\mb{f}h(\mb{x}) +  \alpha(h(\mb{x}))}{(\mathfrak{L}_{L_\mb{f}h} + \mathfrak{L}_{\alpha \circ h})  }\Big\} - 2L\gamma\right)^{2}}\:.
    \end{align} 
Then $h$ is a valid MR-CBF satisfying the conditions of Theorem \ref{thm:safety}.
\end{corollary}

The statement above characterizes a density condition for data used to calibrate complex sensors.
To ensure the feasibility of the MR-CBF procedure, training data should be sampled non-uniformly according to the right hand side of~\eqref{eq:density_bd}. Future work will seek to connect this sampling scheme with dynamically feasible trajectory planning targeted towards data density.

\section{Simulation Results}
\label{sec:result}

In this section we present simulation results using MR-CBFs and data-driven learning models on a simulated robotic Segway platform. 

The Segway can be seen in Figure \ref{fig:cameraFeed_worstCase}, and is modeled with system dynamics derived using the unconstrained Euler-Lagrange Equations. The system is constrained to planar motion by providing identical input torques about both wheels. The resulting degrees of freedom are the Segway's horizontal position $r$, horizontal velocity $\dot{r}$, pitch angle $\theta_y$, and pitch rate $\dot{\theta}_y$. The nominal controller $\mb{k}_d$ is a simple proportional-derivative (PD) controller as in \cite{taylor2020control}. The simulation was written in a Robot Operating System (ROS) based environment using a mixture of C++ and Python to mimic the code structure of the existing hardware platform \cite{gurriet2018towards}.\footnote{The full simulation code
can be found at {\scriptsize \url{https://github.com/rkcosner/cyberpod_sim_ros.git}}
}

The safe set for the simulation was defined as  $\mathcal{C} = \{ \mb{x} \in \R^7 : h_1(\mb{x}) \geq 0,\; h_2(\mb{x})\geq 0 \}  $ with:
\begin{align}
    \label{eqn:extended_CBFs}
    h_1 = - \dot{\theta}_y +  \alpha_e (c - \theta_y + \theta_y^\star) \quad \quad  & \quad \quad h_2 = \dot{\theta}_y + \alpha_e (c + \theta_y - \theta_y^\star) 
\end{align}
where $c\in \R_{++}$, $\alpha_e \in \R_{++}$, and $\theta^\star_y$ is the pitch angle at equilibrium. The MR-OP Filter constraint in (\ref{eqn:MR-OP}) was applied simultaneously to both safety functions (\ref{eqn:extended_CBFs}) and was then implemented using the ECOS Second Order Cone Problem solver \cite{domahidi2013ecos}. The Lipschitz constants in this constraint were estimated by sampling $\mathfrak{L}_{L_\mb{f}h}$, $\mathfrak{L}_{L_\mb{g}h}$ and $\alpha \circ h$ on a set of gridded values around the system's equilibrium point by taking the maximum of the slopes between any two adjacent grid points. As a baseline comparison, a CBF-QP Filter (\ref{eqn:CBF-QP}) was implemented and applied using both safety functions (\ref{eqn:extended_CBFs}). We considered the two following testing scenarios:
\begin{figure*}[t]
    \centering
        \begin{subfloat}
          {\includegraphics[width=0.5\textwidth, valign=t]{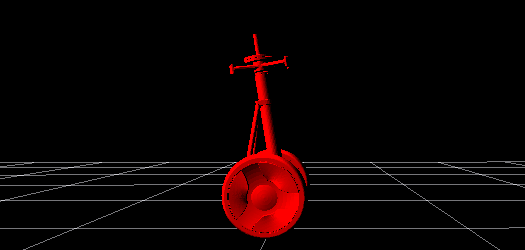}}
          \end{subfloat}
          \begin{subfloat}
          {\includegraphics[width=0.48\textwidth, valign=t]{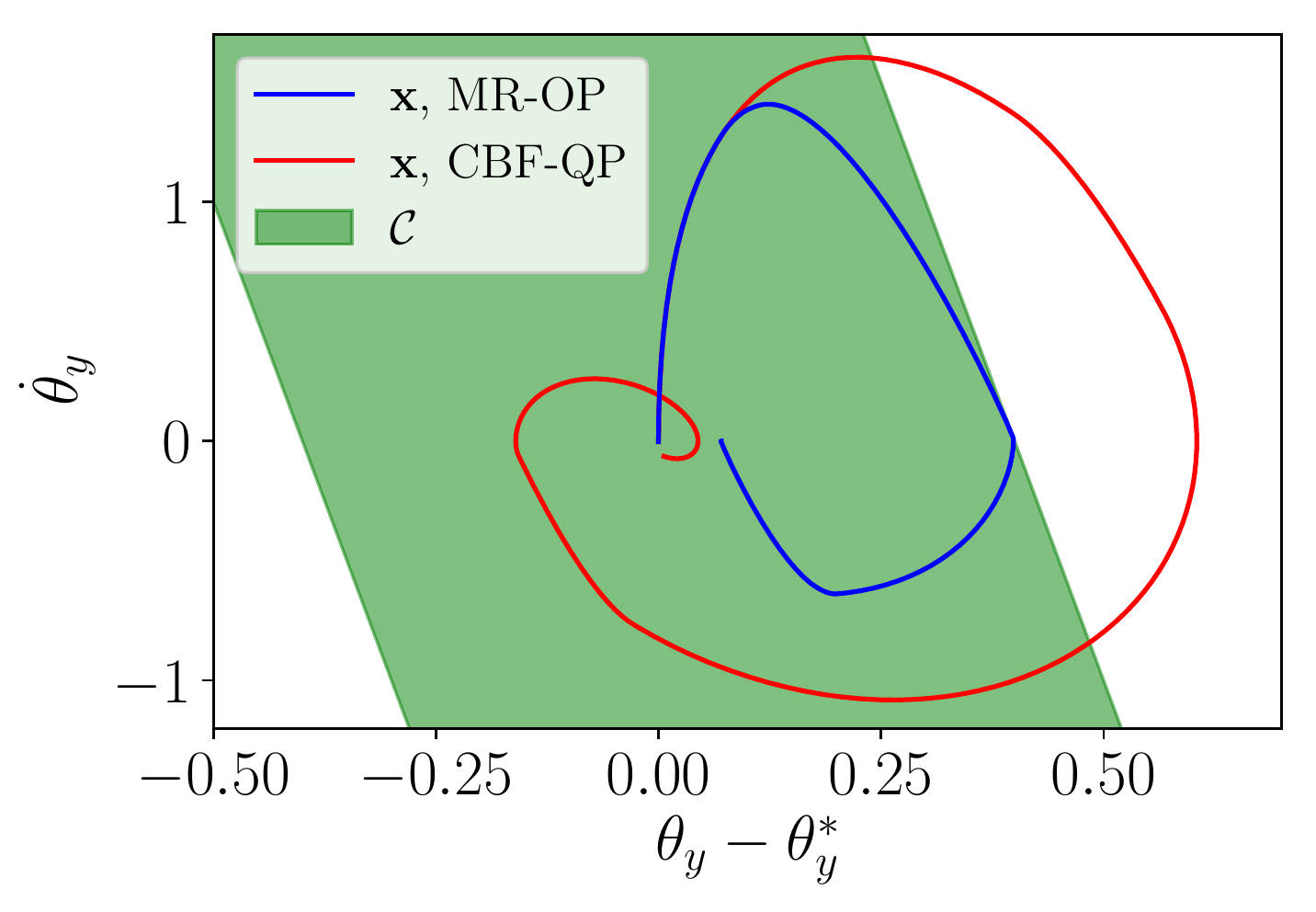}}
          \end{subfloat}
    \caption{ \textbf{(Left)} The Segway model used in simulation from the perspective of the fixed virtual camera used to estimate its state. \textbf{(Right)} Simulation results for worst-case measurement model uncertainty of $\epsilon=0.2$ subtracted from the true pitch angle $\theta_y$ when measured. A state trajectory generated using the the Standard CBF Filter (red) and the MR-OP Filter (blue) are shown as projections onto their pitch angle and pitch rate components. The safe set is plotted in green. Given the same initial condition, the MR-OP filter ensured safety of the trajectory whereas the Standard CBF Filter did not.}
    \label{fig:cameraFeed_worstCase}
\end{figure*}

\paragraph{Worst-Case Synthetic Measurement Model Uncertainty:} In this testing scenario we assumed that direct measurements of the pitch angle $\theta_y$ were offset by a constant factor of $\epsilon>0$, such that $\widehat{\theta}_y=\theta_y-\epsilon$.
Implementing the MR-OP Filter for $\epsilon>0$ ensures safety for this worst-case error of up to $\epsilon$. The result of this type of worst-case measurement model uncertainty in the Segway system with a standard CBF-QP Filter and an MR-OP Filter can be seen in Figure \ref{fig:cameraFeed_worstCase}.

\paragraph{Data-Driven Sensor Calibration:}
In this scenario a more realistic form of measurement model uncertainty is introduced through the use of a learned model to estimate the position $r$ and pitch angle $\theta_y$ from camera images. In simulation, a virtual camera and lighting source were implemented to provide a 15 Hz video feed with a fixed perspective, an example of which can be seen in Figure \ref{fig:cameraFeed_worstCase}. The labels for this supervised-learning problem were noisy measurements of the position and pitch angle generated by the system's inertial measurement unit, corrupted by Gaussian noise with standard deviation $0.1$. We use \texttt{sklearn}'s Kernel Ridge Regression with radial basis functions \citep{pedregosa2011scikit} trained using a set of 800 labelled images associated with a gridded range of position and pitch angle values to ensure dense coverage.
The hyperparameter values $\alpha=0, \gamma=5.4\times 10^{-8}$ were selected to minimize the average error on an 80\%~-~20\% random train-test split. 

The result of the learning-induced errors in the Segway system with a standard CBF Filter and an MR-OP Filter with $\epsilon=0.2$ can be seen in Figure \ref{eqn:fig:cameraFeed_learningCase}.
In the left panel, the safe set $\C$ is shaded according to the upper bound on error $\bar{\varepsilon}(\mb x)$ under which feasibility is guaranteed.
As the errors in the learned map do not exceed this value empirically, the set $\C$ is rendered invariant.
The expression for $\bar{\varepsilon}(\mb x)$ in this experimental scenario is presented in a result analogous to Theorem~\ref{thm:err_bd_x} in Appendix~\ref{app:experiments}, along with an empirical validation of learned model errors.

\begin{figure*}[t]
    \centering
    \begin{subfloat}
    {\includegraphics[width=0.45\textwidth, valign = t]{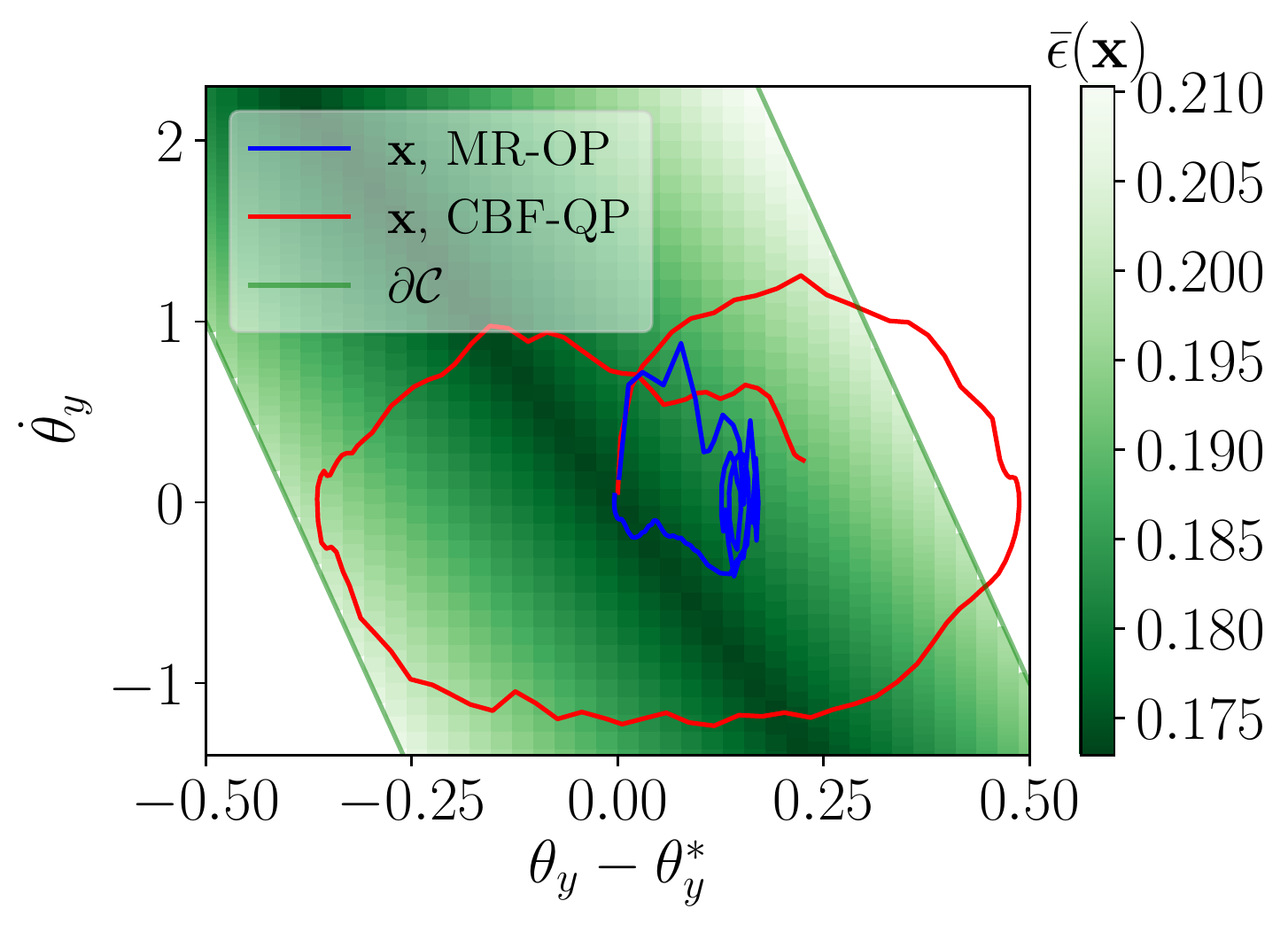}}
    \end{subfloat}
    \begin{subfloat}
    {\includegraphics[width=0.48\textwidth, valign = t]{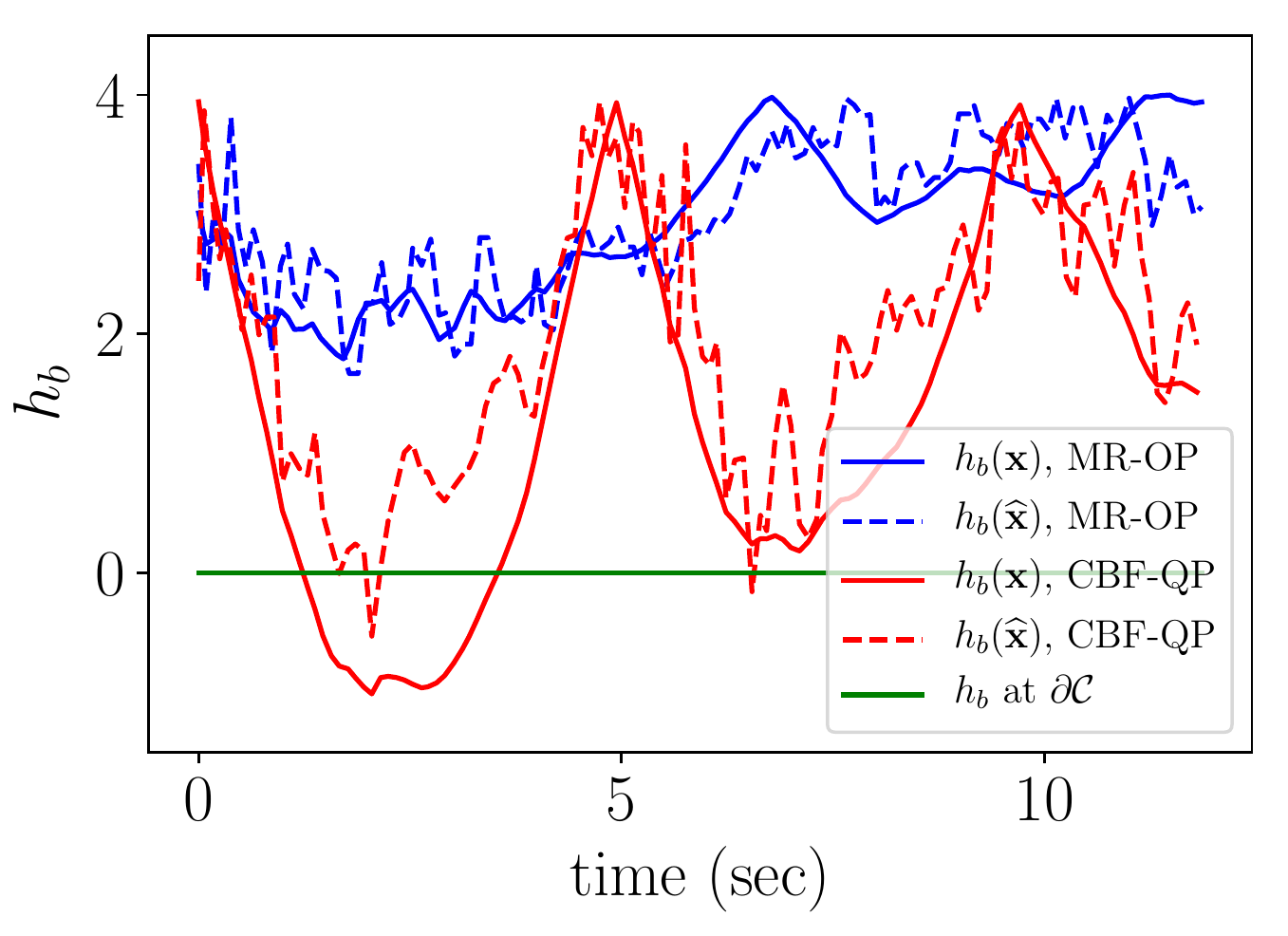}}
    \end{subfloat}
    \caption{Simulation results demonstrating the ability of the MR-OP Filter to mitigate the impact of imperfect learned perception models on safety. \textbf{(Left)} The state trajectory generated using the Standard CBF Filter (red) and MR-OP Filter (blue) are shown as projections onto their pitch angle and pitch rate components. Given the same initial condition, the MR-OP Filter generated a safe trajectory whereas the Standard CBF-QP Filter did not. The state estimates for each trajectory had a maximum error of 0.183 and 0.201 respectively. \textbf{(Right)} The Boolean composition, $h_b = \min \{ h_{e1}, h_{e2}\} = h_{e1} \land h_{e2}$ as defined in \citep{glotfelter2018boolean}, is plotted for the CBF-QP Filter and MR-OP Filter trajectories for the true and estimated states, $\mb{x}$ (solid line) and $\widehat{\mb{x}}$ (dotted line). The safety violation of the Standard CBF-QP Filter can be seen where $h_b(\mb{x})$ crosses 0.}
\label{eqn:fig:cameraFeed_learningCase}
\end{figure*}

\section{Conclusion}
In conclusion, we presented the notion of a Measurement-Robust Control Barrier Function as a tool for ensuring safety in the presence of error in the model relating measurements to state estimates. The resulting safety condition required by a MR-CBF can be directly incorporated into an optimization based controller as a second order cone constraint, preserving the convexity of typical CBF based controllers. We explore how worst case error in learned measurement models can be quantified in terms of data density, and demonstrate how learning-based perception models can be trained on noisy state data and incorporated with MR-CBFs to achieve safe, perception based control.

Future work will seek to explore the concepts in this paper from both a practical and a theoretical standpoint. From a practical perspective, we will demonstrate the feasibility of MR-CBFs and perception-based control on unstable systems such as the Segway in real world experiments. From a theoretical perspective, we will explore the impact of noise, targeted data acquisition through trajectory generation, Lipschitz properties of the resulting optimization based controllers, and strategies for non-invertible observation models such as dynamic estimators (Kalman filters).

\clearpage
\acknowledgments{
We thank the anonymous reviewers for helpful feedback, and Andrew Singletary for his work developing the Segway simulation environment.
This research is generously supported in part by ONR awards N00014-20-1-2497 and N00014-18-1-2833, NSF CPS award 1931853, and the DARPA Assured Autonomy program (FA8750-18-C-0101), and a gift from Twitter.
SD is supported by an NSF Graduate Research Fellowship under Grant No. DGE 1752814. AT is supported by DARPA award HR00111890035.
}

\bibstyle{plain}
\bibliography{taylor_main} 
\clearpage

\appendix
\section{Measurement-Robust Control Barrier Function Safety Proof}
\label{app:mrcbfproof}
In this appendix we provide a proof of Theorem \ref{thm:safety} relating MR-CBFs to safety.

Recall that Theorem \ref{thm:safety} is given by:
\begin{theorem*}[2]
Let a set $\C\subset\R^n$ be defined as the 0-superlevel set of a continuously differentiable function $h:\R^n\to\R$. Assume the functions $L_\mb{f}h:\R^n\to\R$, $L_\mb{g}h:\R^n\to\R^m$, and $\alpha\circ h:\R^n\to\R$ are Lipschitz continuous on $\C$ with Lipschitz coefficients $\mathfrak{L}_{L_\mb{f}h}$, $\mathfrak{L}_{L_\mb{g}h}$, and $\mathfrak{L}_{\alpha\circ h}$, respectively. Further assume there exists a locally Lipschitz function $\epsilon:\R^k\to\R_+$, such that $\max_{\mb{e}\in\mathcal{E}(\mb{y})}\Vert\mb{e}\Vert_2\leq\epsilon(\mb{y})$ for all $\mb{y} \in \mb{p}(\C)$. If $h$ is an MR-CBF for \eqref{eqn:eom} on $\C$ with parameter function $(\epsilon(\mb{y})(\mathfrak{L}_{L_\mb{f}h}+\mathfrak{L}_{\alpha\circ h}),\epsilon(\mb{y})\mathfrak{L}_{L_\mb{g}h})$, then any locally Lipschitz continuous controller $\mb{k}:\R^k\times\R^n\to\R^m$, such that $\mb{k}(\mb{y},\widehat{\mb{x}})\in K_{\textrm{mr-cbf}}(\mb{y},\widehat{\mb{x}})$ for all $(\mb{y},\widehat{\mb{x}})\in\widehat{\mb{v}}(\C)$, renders the system \eqref{eqn:cloop} safe with respect to the set $\C$.
\end{theorem*}

\begin{proof}[Proof of Theorem \ref{thm:safety}]
Define the function $c:\R^n\times\R^m\to\R$ as
    \begin{align*}
        c(\mb{x},\mb{u}) & = \derp{h}{\mb{x}}(\mb{x})\bigg(\mb{f}(\mb{x}) + \mb{g}(\mb{x}) \mb{u}\bigg) + \alpha (h( \mb{x})) ,\\
        & = L_\mb{f}h(\mb{x}) + L_\mb{g}h(\mb{x})\mb{u} + \alpha (h(\mb{x})).
    \end{align*}
This proof will follow from Theorem \ref{thm:cbf_safe}, in that for any $\mb{x}\in\C$, with $(\mb{y},\widehat{\mb{x}})=\widehat{\mb{v}}(\mb{x})$, we will show:
\begin{equation}
\label{eqn:ccond}
    c(\mb{x},\mb{k}(\mb{y},\widehat{\mb{x}}))\geq 0.
\end{equation}

To show that \eqref{eqn:ccond} is true, consider a measurement-state estimate pair $(\mb{y},\widehat{\mb{x}})\in\widehat{\mb{v}}(\C)$. A sufficient condition for \eqref{eqn:ccond} to hold is given by:
\begin{equation*} 
    \inf_{\mb{x}\in\mathcal{X}(\mb{y})}c(\mb{x},\mb{k}(\mb{y},\widehat{\mb{x}})) \geq 0\:.
\end{equation*}
Recalling that we define $\widehat{\mb x} = \mb x + \mb e(\mb x)$, we have:
\begin{align*} 
    \inf_{\mb{x}\in\mathcal{X}(\mb{y})}c(\mb{x},\mb{k}(\mb{y},\widehat{\mb{x}}))  &= \inf_{\mb{e}\in\mathcal{E}(\mb{y})} c(\widehat{\mb x}-\mb e,\mb{k}(\mb{y},\widehat{\mb{x}})),\\
     &= c(\widehat{\mb x},\mb{k}(\mb{y},\widehat{\mb{x}})) + \inf_{\mb{e}\in\mathcal{E}(\mb{y})}
     c(\widehat{\mb x}-\mb e,\mb{k}(\mb{y},\widehat{\mb{x}})) - c(\widehat{\mb x},\mb{k}(\mb{y},\widehat{\mb{x}})),\\
     &\geq c(\widehat{\mb x},\mb{k}(\mb{y},\widehat{\mb{x}})) - \sup_{\mb{e}\in\mathcal{E}(\mb{y})}
     \vert c(\widehat{\mb x}-\mb e,\mb{k}(\mb{y},\widehat{\mb{x}})) - c(\widehat{\mb x},\mb{k}(\mb{y},\widehat{\mb{x}}))\vert.
\end{align*} 
The assumption on Lipschitz continuity of $L_\mb{f}h$, $L_\mb{g}h$, and $\alpha\circ h$ enables the following bound:
    \begin{align*}
    \vert c(\mb{x'},\mb{u}) - c(\mb{x},\mb{u}) \vert & = \vert L_\mb{f}h(\mb{x'})- L_\mb{f}h(\mb{x}) + L_\mb{g}h(\mb{x'})\mb{u}  - L_\mb{g}h(\mb{x})\mb{u} + \alpha(h(\mb{x'})) - \alpha(h(\mb{x}))\vert, \\
    & \leq \vert L_\mb{f}h(\mb{x'})- L_\mb{f}h(\mb{x})\vert + \vert L_\mb{g}h(\mb{x'})\mb{u}  - L_\mb{g}h(\mb{x})\mb{u}\vert + \vert\alpha(h(\mb x')) - \alpha(h(\mb{x}))\vert, \\
    & \leq \mathfrak{L}_{L_\mb{f}h} \lVert\mb{x'} - \mb{x}\rVert_2 + \lVert L_\mb{g}h(\mb{x'}) - L_\mb{g}h(\mb{x})\rVert_2 \lVert\mb{u}\rVert_2 + \mathfrak{L}_{\alpha \circ h} \lVert\mb{x' } - \mb{x}\rVert_2,\\
    & \leq (\mathfrak{L}_{L_\mb{f}h} + \mathfrak{L}_{L_\mb{g}h}\lVert\mb{u}\rVert_2 + \mathfrak{L}_{\alpha \circ h}) \lVert\mb{x'} - \mb{x}\rVert_2.
    \end{align*}
Therefore, using the definition of $\epsilon(\mb y)$ we have:
    \begin{align*}
        \sup_{\mb{e}\in\mathcal{E}(\mb{y})}
     \vert c(\widehat{\mb x}-\mb e,\mb{k}(\mb{y},\widehat{\mb{x}})) - c(\widehat{\mb x},\mb{k}(\mb{y},\widehat{\mb{x}}))\vert &\leq 
     \sup_{\mb{e}\in\mathcal{E}(\mb{y})}  (\mathfrak{L}_{L_\mb{f}h} + \mathfrak{L}_{L_\mb{g}h}\lVert\mb{k}(\mb{y},\widehat{\mb{x}})\rVert_2 + \mathfrak{L}_{\alpha \circ h}) \lVert\mb e\rVert_2, \\
     &\leq (\mathfrak{L}_{L_\mb{f}h} + \mathfrak{L}_{L_\mb{g}h}\lVert\mb{k}(\mb{y},\widehat{\mb{x}})\rVert_2 + \mathfrak{L}_{\alpha \circ h}) \epsilon(\mb y).
    \end{align*}
Thus:
\begin{align*} 
    \inf_{\mb{x}\in\mathcal{X}(\mb{y})}c(\mb{x},\mb{k}(\mb{y},\widehat{\mb{x}}))  \geq c(\widehat{\mb x},\mb{k}(\mb{y},\widehat{\mb{x}})) - (\mathfrak{L}_{L_\mb{f}h} + \mathfrak{L}_{L_\mb{g}h}\lVert\mb{k}(\mb{y},\widehat{\mb{x}})\rVert_2 + \mathfrak{L}_{\alpha \circ h}) \epsilon(\mb y).
\end{align*} 
By the MR-CBF condition and the design of $\mb{k}$ we have that:
    \begin{align*}
        c(\widehat{\mb x},\mb{k}(\mb{y},\widehat{\mb{x}})) - \epsilon(\mb y)(\mathfrak{L}_{L_\mb{f}h} + \mathfrak{L}_{\alpha \circ h})  - \epsilon(\mb y)\mathfrak{L}_{L_\mb{g}h}\lVert\mb{k}(\mb{y},\widehat{\mb{x}})\rVert_2 \geq 0, 
    \end{align*}
implying the condition \eqref{eqn:ccond}.
  
\end{proof}

\section{MR-OP Second Order Cone Conversion}
\label{app:mropsocp}
In this appendix we show how the optimization problem specified by \eqref{eqn:MR-OP} can be written in the standard form for a Second-Order Cone Program.
The MR-OP Filter is given by:
\begin{align}
    \tag{MR-OP}
    \mb{k}(\mb{y},\widehat{\mb{x}}) =  \,\,\underset{\mb{u} \in \R^m}{\argmin}  &  \quad \frac{1}{2} \Vert \mb{u}-\mb{k}_d(\widehat{\mb{x}}) \Vert_2^2  \\
    \mathrm{s.t.} \quad & \quad L_\mb{f}h(\widehat{\mb{x}})-(\mathfrak{L}_{L_\mb{f}h}+\mathfrak{L}_{\alpha\circ h})\epsilon(\mb{y})+L_\mb{g}h(\widehat{\mb{x}})\mb{u}-\mathfrak{L}_{L_\mb{g}h}\epsilon(\mb{y})\Vert\mb{u}\Vert_2\geq-\alpha(h(\widehat{\mb{x}})). \nonumber
\end{align}
First, the constant term can be removed from the cost such that it becomes: $\frac{1}{2}\Vert\mb{u} \Vert_2^2 - \mb{k}_d(\widehat{\mb{x}})^\top \mb{u} $. Additionally, the constraint can be written in terms of a second order cone:  
\begin{equation}
\quad \mb{q} = \lmat L_\mb{g}h(\widehat{\mb{x}})  \\ \mathfrak{L}_{L_\mb{g}h}\epsilon(\mb{y})  I_m \rmat \mb{u}   + \lmat \alpha(h(\widehat{\mb{x}})) + L_\mb{f}h(\widehat{\mb{x}})-(\mathfrak{L}_{L_\mb{f}h}+\mathfrak{L}_{\alpha\circ h})\epsilon(\mb{y})\\ 0_m %
\rmat \nonumber,  \quad \mb{q} \in Q^{m+1}, 
\end{equation}
where $Q^{m+1} \triangleq \{ (q_0, \mb{q}_1) \in \R\times \R^{m} ~|~ q_0 \geq \Vert\mb{q}_1\Vert_2 \} $ is the second-order cone in $\R^{m+1}$. Furthermore, by adding the decision variable $t \in \R$, the quadratic cost function can be converted to an equivalent linear cost, $t -  \mb{k}_d(\widehat{\mb{x}})^T\mb{u}$, with a rotated second order cone constraint, $\Vert \mb{u}\Vert^2_2 \leq 2t$, that can be converted to a standard second order cone constraint via a rotation matrix $\mb{R} \in \R^{(m+2)\times (m+2)}$ \citep{calafiore2014optimization}. Thus the additional constraint can be written as:
\begin{equation}
    \mb{r} = \mb{R}\lmat t & 1 & \mb{u}^\top \rmat^\top  \nonumber, \;   \mb{r} \in Q^{m+2}.  \\
\end{equation}
Combining the two second order cone constraints, the problem can be written as:
\begin{align}
    \mb{k}(\mb{y},\widehat{\mb{x}}) =   \,\,\underset{(\mb{u},t,\mb{s}) \in \R^{3m+4}}{\argmin} &  \quad \lmat 1 & 0 & -\mb{k}_d(\widehat{\mb{x}})^\top \rmat \lmat t & 1 & \mb{u}^\top \rmat^\top \nonumber\\
     \mathrm{s.t.}&  \quad \mb{G}(\mb{y},\widehat{\mb{x}})   \lmat t \\1 \\ \mb{u} \rmat  + \mb{s} = \mb{h}(\mb{y}, \widehat{\mb{x}}),  \nonumber\\
    & \quad \mb{s} \in Q^{m+2} \times Q^{m+1}, \nonumber
\end{align}
where:
\begin{align*}
    \mb{G}(\mb{y},\widehat{\mb{x}}) &= - \lmat  \mb{R}_1 & 0 &  \mb{R}_{3:m+2} \\   0 & 0 & L_\mb{g}h(\widehat{\mb{x}}) \\ \mb{0}_m & \mb{0}_m & \mathfrak{L}_{L_\mb{g}h}\epsilon(\mb{y}) \mb{I}_m  \rmat, \\ \mb{h}(\mb{y},\widehat{\mb{x}}) & = \lmat \mb{R}_2\\ \alpha(h(\widehat{\mb{x}})) + L_\mb{f}h(\widehat{\mb{x}})-(\mathfrak{L}_{L_\mb{f}h}+\mathfrak{L}_{\alpha\circ h})\epsilon(\mb{y})\\ \mb{0}_m     \rmat.
\end{align*}
Here $\mb{R}_i$ represents the $i^\textrm{th}$ column of $\mb{R}$, and $\mb{R}_{3:m+2}$ represents the columns 3 through $m+2$. This equivalent formulation of (\ref{eqn:MR-OP}) is in standard Second-Order Cone Program form as in \citep{domahidi2013ecos}.

\section{R-MR-OP Lipschitz Continuity Proof}
\label{app:rmroplipschitz}
In this appendix we show that the controller \eqref{eqn:R-MR-OP} is locally Lipschitz continuous in terms of the true state $\mb{x}$. Recall this controller is given by:
\begin{align}
\tag{R-MR-OP}
\mb{k}(\mb{y},\widehat{\mb{x}}) =  \,\,\underset{(\mb{u},\delta) \in \R^m\times\R}{\argmin}  &  \quad \frac{1}{2} \Vert \mb{u}-\mb{k}_d(\widehat{\mb{x}}) \Vert_2^2 + p\delta^2  \\
\mathrm{s.t.} \quad  L_\mb{f}h&(\widehat{\mb{x}})-(\mathfrak{L}_{L_\mb{f}h}+\mathfrak{L}_{\alpha\circ h})\epsilon(\mb{y})+L_\mb{g}h(\widehat{\mb{x}})\mb{u}-\mathfrak{L}_{L_\mb{g}h}\epsilon(\mb{y})\Vert\mb{u}\Vert_2\geq-\alpha(h(\widehat{\mb{x}}))-\delta. \nonumber
\end{align}
where $\mb{y}=\mb{p}(\mb{x})$ and $\widehat{\mb{q}}(\mb{p}(\mb{x}))$, with $\mb{p}$ and $\widehat{\mb{q}}$ locally Lipschitz continuous. Thus if we prove that $\mb{k}$ is locally Lipschitz with respects to its arguments, the composition of $\mb{k}$ with the measurement functions will be locally Lipschitz. We also note that this optimization is always feasible, with a unique minimizer as it is convex.

The constraint in this optimization problem can be restated as a conic constraint:
\begin{align}
\tag{R-MR-OP}
\mb{k}(\mb{y},\widehat{\mb{x}}) =  \,\,\underset{(\mb{u},\delta) \in \R^m\times\R}{\argmin}  &  \quad \frac{1}{2} \Vert \mb{u}-\mb{k}_d(\widehat{\mb{x}}) \Vert_2^2 + p\delta^2  \\
\mathrm{s.t.} \quad & \quad \begin{bmatrix}L_\mb{g}h(\widehat{\mb{x}})\mb{u} + \delta + \alpha(h(\widehat{\mb{x}})) + L_\mb{f}h(\widehat{\mb{x}})-(\mathfrak{L}_{L_\mb{f}h}+\mathfrak{L}_{\alpha\circ h})\epsilon(\mb{y}) \\ \mathfrak{L}_{L_\mb{g}h}\epsilon(\mb{y})\mb{u} \end{bmatrix} \in Q^{m+1}. \nonumber
\end{align}
To prove local Lipschitz continuity of this controller, we will make use of Theorem 4 in \cite{dontchev1993lipschitzian}, stated below. To draw parallels with the notation of this work, we define the following:
\begin{align}
\mb{w} &\triangleq (\mb{u},\delta) \in\R^{m+1}, \\
C_{(\mb{y},\widehat{\mb{x}})}(\mb{w}) 
&\triangleq  \frac{1}{2} \Vert \mb{u}-\mb{k}_d(\widehat{\mb{x}}) \Vert_2^2 + p\delta^2, \label{eqn:C} \\
G_{(\mb{y},\widehat{\mb{x}})}(\mb{w}) &\triangleq \begin{bmatrix}L_\mb{g}h(\widehat{\mb{x}})\mb{u} + \delta + \alpha(h(\widehat{\mb{x}})) + L_\mb{f}h(\widehat{\mb{x}})-(\mathfrak{L}_{L_\mb{f}h}+\mathfrak{L}_{\alpha\circ h})\epsilon(\mb{y}), \\ \mathfrak{L}_{L_\mb{g}h}\epsilon(\mb{y})\mb{u} \end{bmatrix}      \label{eqn:G}
\end{align}
The optimization problem can then be written as:
\begin{align}
\label{eqn:CGRMROP}
\mb{k}(\mb{y},\widehat{\mb{x}}) =  \,\,\underset{\mb{w} \in \R^{m+1}}{\argmin}  &  \quad C_{(\mb{y},\widehat{\mb{x}})}(\mb{w})  \\
\mathrm{s.t.} \quad & \quad G_{(\mb{y},\widehat{\mb{x}})}(\mb{w}) \in Q^{m+1}. \nonumber
\end{align}
We note that $G_{(\mb{y},\widehat{\mb{x}})}:\R^{m+1}\to\R^{m+1}$, with $\R^{m+1}$ a Banach space, and that $Q^{m+1}\subset\R^{m+1}$ is a closed, convex cone with its vertex at the origin. We further have that $C_{(\mb{y},\widehat{\mb{x}})}$ and $G_{(\mb{y},\widehat{\mb{x}})}$ are twice differentiable with respect to $\mb{w}$, and these derivatives are locally Lipschitz continuous with respect to $(\mb{y},\widehat{\mb{x}})$ by local Lipschitz continuity of $\mb{k}_d$, $L_\mb{f}h$, $L_\mb{g}h$, $\alpha\circ h$, and $\epsilon$.

The $H_{(\mb{y},\widehat{\mb{x}})}$ denote the Lagrangian:
\begin{equation}
    \label{eqn:H}
H_{(\mb{y},\widehat{\mb{x}})}(\mb{w},\bs{\lambda}) = C_{(\mb{y},\widehat{\mb{x}})}(\mb{w})-\langle\bs{\lambda},G_{(\mb{y},\widehat{\mb{x}})}(\mb{w})\rangle,
\end{equation}
where $\bs{\lambda}\in\R^m$. The first-order necessary conditions associated with a solution $(\mb{w}^\star_{(\mb{y},\widehat{\mb{x}})},\bs{\lambda}^\star_{(\mb{y},\widehat{\mb{x}})})$ to \eqref{eqn:CGRMROP} can be expressed as:
\begin{align}
  \label{eqn:foc}
  \derp{H_{(\mb{y},\widehat{\mb{x}})}}{\mb{w}}(\mb{w}^\star_{(\mb{y},\widehat{\mb{x}})},\bs{\lambda}^\star_{(\mb{y},\widehat{\mb{x}})}) & = \mb{0}, \\
  \langle \bs{\lambda}_{(\mb{y},\widehat{\mb{x}})}^\star,G_{(\mb{y},\widehat{\mb{x}})}(\mb{w}^\star_{(\mb{y},\widehat{\mb{x}})}) \rangle & = 0, \\
  G_{(\mb{y},\widehat{\mb{x}})}(\mb{w}^\star_{(\mb{y},\widehat{\mb{x}})}) &\in  Q^{m+1^+},
\end{align}
where $\mb{w}^\star_{(\mb{y},\widehat{\mb{x}})}\in\R^m+1$ and $\bs{\lambda}^\star_{(\mb{y},\widehat{\mb{x}})}\in Q^{m+1}$. We also have that the following coercivity condition is met:
\begin{equation}
\label{eqn:coerc}
    \left\langle \derp{H_{(\mb{y},\widehat{\mb{x}})}}{\mb{w}}(\mb{w}^\star_{(\mb{y},\widehat{\mb{x}})},\bs{\lambda}^\star_{(\mb{y},\widehat{\mb{x}})})(\mb{v}_2-\mb{v}_1),\mb{v}_2-\mb{v}_1 \right\rangle \geq \Vert\mb{v}_2-\mb{v}_1\Vert_2^2,
\end{equation}
for any $\mb{v}_1,\mb{v}_2\in\R^{m+1}$ as:
\begin{equation}
    \derp{H_{(\mb{y},\widehat{\mb{x}})}}{\mb{w}}(\mb{w}^\star_{(\mb{y},\widehat{\mb{x}})},\bs{\lambda}^\star_{(\mb{y},\widehat{\mb{x}})}) = \mb{I}_{m+1},
\end{equation}
for all $(\mb{y},\widehat{\mb{x}})$. Lastly, we note that:
\begin{equation}
\label{eqn:surj}
    \derp{G_{(\mb{y},\widehat{\mb{x}})}}{\mb{w}}(\mb{w}^\star_{(\mb{y},\widehat{\mb{x}})}) = \begin{bmatrix} L_\mb{g}h(\widehat{\mb{x}}) & 1 \\ \mathfrak{L}_{L_\mb{g}h}\epsilon(\mb{y})\mb{I}_{m} & \mb{0}_{m\times 1} \end{bmatrix}.
\end{equation}
Under the assumption that $\epsilon(\mb{y})\neq 0$ for all $\mb{y}\in\mb{p}(\C)$ (or that every measurement has some amount of worst case error), we have the map in  $\derp{G_{(\mb{y},\widehat{\mb{x}})}}{\mb{w}}(\mb{w}^\star_{(\mb{y},\widehat{\mb{x}})})$ is surjective from $\R^{m+1}$ to $\R^{m+1}$. Thus we have that all of the conditions of the following theorem are met:
\begin{theorem}[Lipschitz Continuity of SOCP \cite{dontchev1993lipschitzian}]
If ~$\derp{G_{(\mb{y},\widehat{\mb{x}})}}{\mb{w}}(\mb{w}^\star_{(\mb{y},\widehat{\mb{x}})})$ is surjective and the coercivity condition \eqref{eqn:coerc} holds, then there exists $s>0$ such that \eqref{eqn:CGRMROP} has a strict local minimizer $\mb{w}^\star_ {(\mb{y},\widehat{\mb{x}})}$ for each $(\mb{y}',\widehat{\mb{x}}')\in B_s((\mb{y},\widehat{\mb{x}}))$, and both $\mb{w}^\star_{(\mb{y},\widehat{\mb{x}})}$, and the associated (unique) multiplier $\bs{\lambda}^\star_{(\mb{y},\widehat{\mb{x}})}$ satisfying the first-order necessary condition \eqref{eqn:foc}, are Lipschitz continuous functions of $(\mb{y}',\widehat{\mb{x}}')\in B_s((\mb{y},\widehat{\mb{x}}))$.
\end{theorem}
As $(\mb{y},\widehat{\mb{x}})$ were arbitrary, and are locally Lipschitz continuous functions of the true state $\mb{x}$, we have that $\mb{k}$ is locally Lipschitz with respect to the state.

\section{Learning and Uncertainty Reduction Results}
\label{app:feasibleproof}

In this section, we provide a proof of Theorem~\ref{thm:err_bd_x}, give a motivation for the definition of the nonparametric error bound, and provide a proof of Corollary~\ref{coro:dense_data}.

Recall that Theorem~\ref{thm:err_bd_x} is given by:
\begin{theorem*}[\ref{thm:err_bd_x}]
Let $\C$, $h$, $\mathfrak{L}_{L_\mb{f}h}$, $\mathfrak{L}_{L_\mb{g}h}$, $\mathfrak{L}_{\alpha\circ h}$, and $\epsilon$ be defined as in Theorem~\ref{thm:safety}.
Then $h$ is a MR-CBF for \eqref{eqn:eom} on $\C$ with parameter function $(\epsilon(\mb y)(\mathfrak{L}_{L_\mb{f}h}+\mathfrak{L}_{\alpha\circ h}),\epsilon(\mb{y})\mathfrak{L}_{L_\mb{g}h})$ if for all $\mb x \in \C$:
\begin{equation*}
      \varepsilon(\mb{x})< \max\left\{
      \frac{\Vert L_\mb{g}h({\mb{ x}})\Vert_2}{2\mathfrak{L}_{L_\mb{g}h}},
      \frac{L_\mb{f}h({\mb{x}}) +  \alpha(h({\mb{ x}}))}{2(\mathfrak{L}_{L_\mb{f}h} + \mathfrak{L}_{\alpha \circ h})}
      \right\} \triangleq \bar{\varepsilon}(\mb x)\:.
\end{equation*}
\end{theorem*}

\begin{proof}
For completeness, we begin by proving the condition~\eqref{eq:esp_y_bd}.
The function $h$ is a MR-CBF if
\begin{equation}\label{eq:mr_cbf_cond_app}
     \sup_{\mb{u}\in\R^m} L_\mb{g}h(\widehat{\mb{x}})\mb{u} - b(\mb{y}) \Vert\mb{u}\Vert_2 >-L_\mb{f}h(\widehat{\mb{x}})-\alpha(h(\widehat{\mb{x}})) + a(\mb{y}).
\end{equation}
Considering the left hand side, and doing a change of variables with $s=\|\mb u\|_2$ and $\mb v = \mb u/\|\mb u\|_2$,
\begin{align*}
    \sup_{\mb{u}\in\R^m} L_\mb{g}h(\widehat{\mb{x}})\mb{u} - b(\mb{y}) \Vert\mb{u}\Vert_2 &= \sup_{s>0}
    s\cdot \left(\max_{\Vert\mb v\Vert_2 = 1} L_\mb{g}h(\widehat{\mb{x}})\mb v-b(\mb{y}) \right), \\
    &= \sup_{s>0}
    s\cdot \left(
    \Vert L_\mb{g}h(\widehat{\mb{x}})\Vert_2
    -b(\mb{y}) 
    \right).
\end{align*}
If $\Vert L_\mb{g}h(\widehat{\mb{x}})\Vert_2
    -b(\mb{y}) > 0$, then $s$ can be chosen arbitrarily large, implying~\eqref{eq:mr_cbf_cond_app} holds. If $\Vert L_\mb{g}h(\widehat{\mb{x}})\Vert_2
    -b(\mb{y}) \leq 0$, then \eqref{eq:mr_cbf_cond_app} only holds if
    \[0 >-L_\mb{f}h(\widehat{\mb{x}})-\alpha(h(\widehat{\mb{x}})) + a(\mb{y}).\]
Substituting $a(\mb y)=\epsilon(\mb{y})\mathfrak{L}_{L_\mb{g}h}$ and $b(\mb y)=\epsilon(\mb{y})(\mathfrak{L}_{L_\mb{f}h} + \mathfrak{L}_{\alpha \circ h})$ and combining these conditions, we see that \eqref{eq:mr_cbf_cond_app} is satisfied if and only if

\begin{equation*}
      \epsilon(\mb{y}) = \varepsilon(\mb x) < \max\left\{
      \frac{\Vert L_\mb{g}h(\widehat{\mb{ x}})\Vert_2}{\mathfrak{L}_{L_\mb{g}h}},
      \frac{L_\mb{f}h(\widehat{\mb{x}}) +  \alpha(h(\widehat{\mb{ x}}))}{(\mathfrak{L}_{L_\mb{f}h} + \mathfrak{L}_{\alpha \circ h})}
      \right\}\:.
\end{equation*}

It remains to find a sufficient condition in which the right hand side depends only on $\mb x$.
Notice that
\begin{align*}
    \Vert L_\mb{g}h(\mb{x})\Vert_2 = \Vert L_\mb{g}h(\widehat{\mb{x}} - \mb{e}(\mb x))\Vert_2 \leq \Vert L_\mb{g}h(\widehat{\mb{x}})\Vert_2 + \mathfrak{L}_{L_\mb{g}h}\epsilon,
\end{align*}
therefore,
\begin{align*}
    \Vert L_\mb{g}h(\mb{x})\Vert_2 > 2 \mathfrak{L}_{L_\mb{g}h}\varepsilon(\mb x) \implies 
    \Vert L_\mb{g}h(\widehat{\mb{x}})\Vert_2 >  \mathfrak{L}_{L_\mb{g}h}\varepsilon(\mb x).
\end{align*}
By a similar argument, we also have that
\[
L_\mb{f}h({\mb{x}}) +  \alpha(h({\mb{ x}})) >  2(\mathfrak{L}_{L_\mb{f}h} + \mathfrak{L}_{\alpha \circ h})\varepsilon(\mb x)
\implies 
L_\mb{f}h(\widehat{\mb{x}}) +  \alpha(h(\widehat{\mb{ x}})) >  (\mathfrak{L}_{L_\mb{f}h} + \mathfrak{L}_{\alpha \circ h})\varepsilon(\mb x).
\]
Combining these two expressions gives a sufficient condition depending only on $\mb x$ and yields the desired result.
\end{proof}

\subsection{Motivation for Non-Parametric Error Bound}

Consider the simple Nadarya-Watson non-parametric regressor defined from a set of training data $\mathcal S =\{(\mb{y}_i, \bar{\mb{x}}_i)\}_{i=1}^N$
 as
\begin{align}
\begin{split}\label{eq:regression_estimator}
\widehat{\mb q}(\mb y) = \sum_{i=1}^N \frac{\mathbf{1}\{\rho(\mb y_i,\mb y)\leq\alpha\}}{s_N(\mb y)} \cdot \bar{\mb x}_i ,\quad s_N(\mb y) = \sum_{i=1}^N  \mathbf{1}\{\rho(\mb y_i,\mb y)\leq\alpha\}\:,
\end{split}
\end{align}
where $\rho:\R^k\times \R^k\to \R_+$ is a metric with respect to which the functions $\mb p$ and $\mb q$ are smooth:
\begin{align}
\begin{split}\label{eq:smoothness_rho}
 \rho(\mb p(\mb x), \mb p(\mb x'))\leq \mathfrak{L}_{\mb p} \|\mb x - \mb x'\|_2,
 \quad 
 \|\mb q(\mb y) - \mb q(\mb y')\|_2 \leq \mathfrak{L}_{\mb q} \rho(\mb y, \mb y')\:.
\end{split}
\end{align}
This regressor performs local averaging over training data to predict the underlying state as a function of the system output $\mb y$.

\begin{lemma}[Adapted from~\cite{dean2020robust}]\label{lem:nw_err_nd}
For a system satisfying~\eqref{eq:smoothness_rho} and a learned perception map of the form~\eqref{eq:regression_estimator} with noisy labels as in~\eqref{eqn:noisy_x} with $\mathbb{E}[\mb w_i] = 0$ and $\mathbb{E}[\mb w_i \mb w_i^\top] = \sigma_w^2$, we have with probability at least $1-\delta$ that for any fixed $\mb y$ with $s_N(\mb y)\neq 0$,
\begin{align}
\|\widehat{\mb q}(\mb y) - {\mb q}(\mb y)\|_\infty &\leq \alpha \mathfrak{L}_{\mb q} +  \frac{n\sigma_w}{\sqrt{s_N(\mb y)}} \sqrt{n\log\left(n\sqrt{s_N(\mb y)}/\delta\right)}\:.
\end{align}
\end{lemma}

To see how this applies to the error bound presented in Definition~\ref{def:nonparametric}, we first note that
\begin{align*}
    s_N(\mb y) &= \sum_{i=1}^N  \mathbf{1}\{\rho(\mb p(\mb x_i),\mb p(\mb x))\leq\alpha\},\\
     &\geq \sum_{i=1}^N  \mathbf{1}\{\|\mb x_i-\mb x\|_2\leq\tfrac{\alpha}{\mathfrak{L}_{\mb{p}}}\},\\
     &=\#\{\mb x_i \in \mathcal X_{\mathcal S} \mid \Vert\mb{x} - \mb{x}_i\Vert\leq \tfrac{\alpha}{\mathfrak{L}_{\mb{p}}}\}.
\end{align*}

Therefore, Lemma~\ref{lem:nw_err_nd} implies a pointwise high probability bound on error of the form in Definition~\ref{def:nonparametric} with 
\[L = \mathfrak{L}_{\mb{p}} \mathfrak{L}_{\mb{q}},\quad  \sigma=n\sigma_w \sqrt{n\log(n\sqrt{N}/\delta)},\quad  \gamma=\tfrac{\alpha}{\mathfrak{L}_{\mb{p}}}\:.\]

Extending this probabilistic result from pointwise to hold over a compact set is nontrivial.
The subtleties involved are handled more carefully in~\cite{dean2020robust}, and we focus only on this illustrative connection here.

\subsection{Proof of Corollary~\ref{coro:dense_data}}
This result follows by combining the expression~\eqref{eq:upper_eps_x} with the nonparametric error bound introduced in Definition~\ref{def:nonparametric}.
Using this definition, we have that
\begin{equation*}
      \varepsilon(\mb{x}) = \sup_{\mb e \in\mathcal E(\mb y)} \|\mb e\|_2 \leq L\gamma +  \frac{\sigma}{\sqrt{\#\{\mb x_i \in \mathcal X_{\mathcal S} \mid \Vert\mb{x} - \mb{x}_i\Vert\leq \gamma\}}}\:.
\end{equation*}
Combining this with the sufficient condition,
\begin{equation*}
      \frac{\sigma}{\sqrt{\#\{\mb x_i \in \mathcal X_{\mathcal S} \mid \Vert\mb{x} - \mb{x}_i\Vert\leq \gamma\}}} < \max\left\{\frac{\Vert L_\mb{g}h({\mb{ x}})\Vert_2}{2\mathfrak{L}_{L_\mb{g}h}},
      \frac{L_\mb{f}h({\mb{x}}) +  \alpha(h({\mb{ x}}))}{2(\mathfrak{L}_{L_\mb{f}h} + \mathfrak{L}_{\alpha \circ h})}
      \right\} -L\gamma   \:.
\end{equation*}
The result follows by rearranging terms.

 \section{Further Experimental Details}\label{app:experiments}
 In this section, we provide additional plots related to the learned model and present a result analogous to Theorem~\ref{thm:err_bd_x} specialized to the experimental setting.
 
An empirical validation of the learned model is presented in Figure~\ref{fig:learning}, which displays the training data, the model predictions of $\theta_y$, and the errors over the validation set.

\begin{figure*}
    \centering
    \begin{subfloat}
    {\includegraphics[height=4cm]{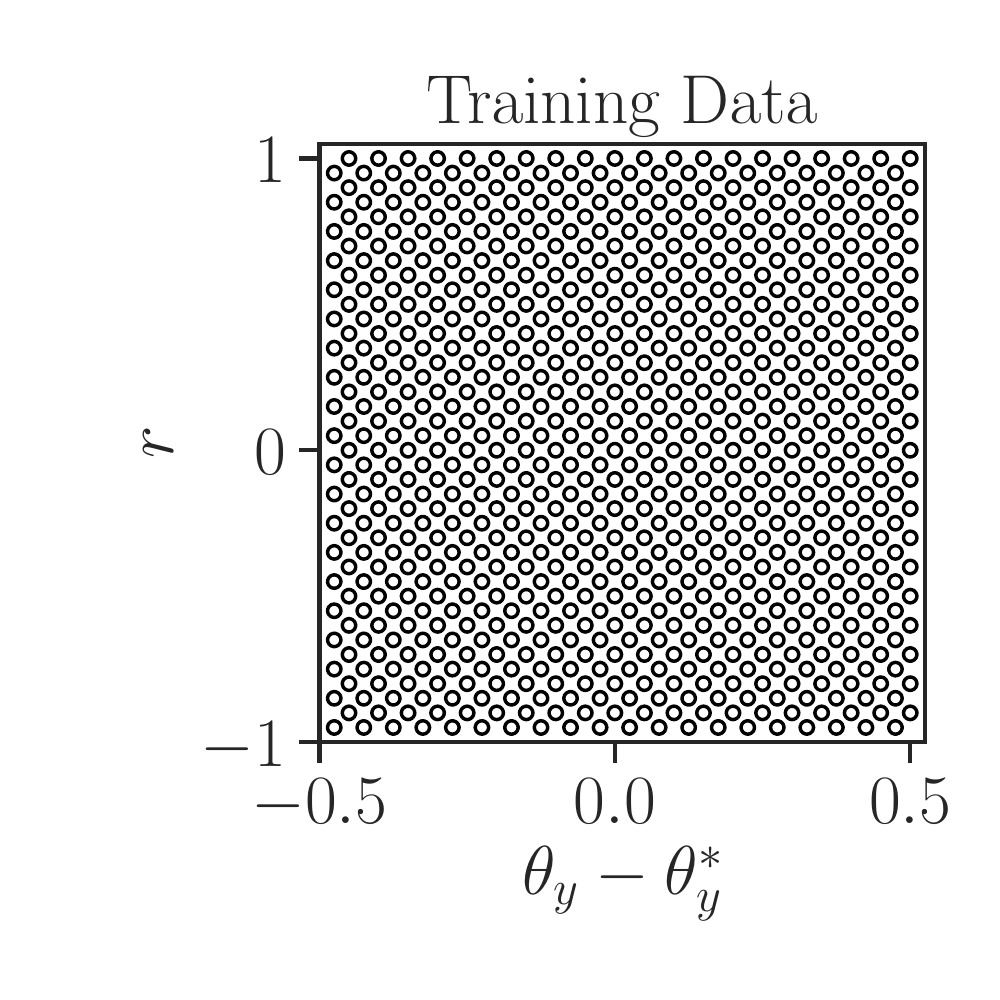}}
    \end{subfloat}
    \begin{subfloat}
    {\includegraphics[height=4cm]{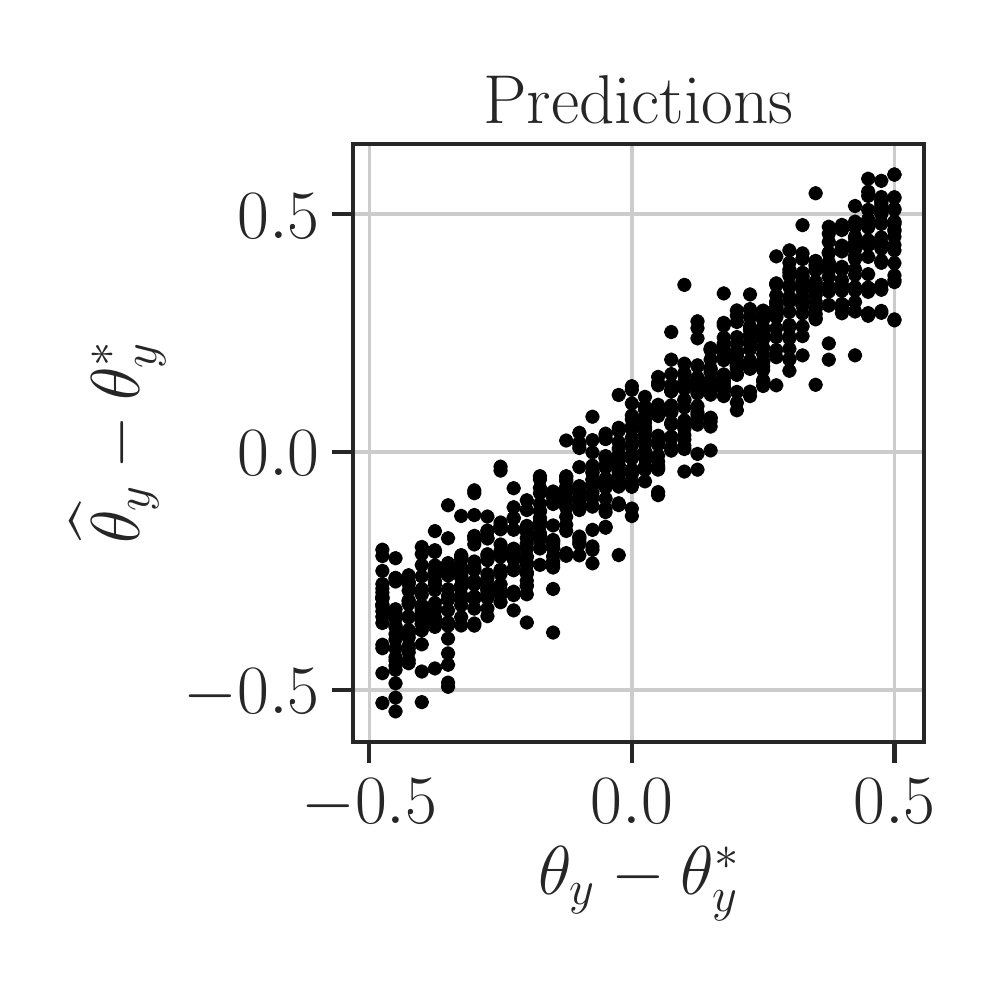}}
    \end{subfloat}
    \begin{subfloat}
    {\includegraphics[height=4cm]{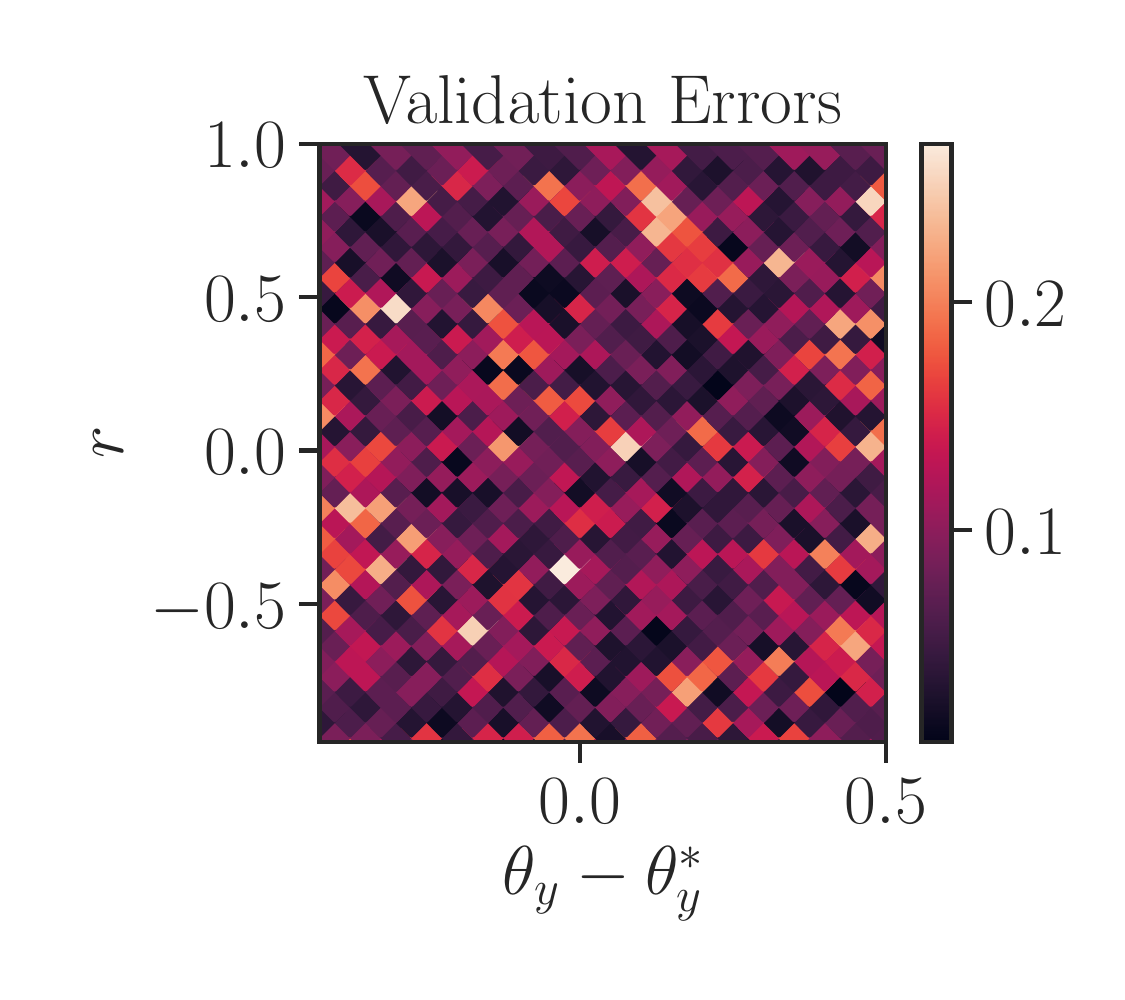}}
    \end{subfloat}
    \caption{Learned model results. (\textbf{Left}) Illustration of the training data, which was collected from a gridded range of position $r$ and pitch angle $\theta_y$. (\textbf{Middle}) The estimated vs. actual pitch angles on a validation set. (\textbf{Right}) The empirical model errors $\varepsilon(\mb x)$ measured on a validation set.}
\label{fig:learning}
\end{figure*}
 
We now present a result analogous to Theorem~\ref{thm:err_bd_x} to derive an expression for $\bar{\varepsilon}(\mb x)$ in the experimental scenario.
For control of the robotic Segway under planar motion,
the input torques about both wheels were constrained to be identical.
Furthermore, the MR-OP Filter constraint in (\ref{eqn:MR-OP}) was applied simultaneously to both safety functions (\ref{eqn:extended_CBFs}).
Therefore, the MR-OP Filter requires the feasibility of the following set of constraints, for $u\in\R$:
\begin{align*}
& \mb u = u \mb 1,\\
& L_\mb{f}h_1(\widehat{\mb{x}})+L_\mb{g}h_1(\widehat{\mb{x}})\mb{u}
-(\mathfrak{L}_{L_\mb{f}h_1}+\mathfrak{L}_{\alpha\circ h_1})\epsilon(\mb{y}) - \mathfrak{L}_{L_\mb{g}h_1}\epsilon(\mb{y})\Vert\mb{u}\Vert_2
\geq-\alpha(h_1(\widehat{\mb{x}})),\\
& L_\mb{f}h_2(\widehat{\mb{x}})+L_\mb{g}h_2(\widehat{\mb{x}})\mb{u}
-(\mathfrak{L}_{L_\mb{f}h_2}+\mathfrak{L}_{\alpha\circ h_2})\epsilon(\mb{y}) - \mathfrak{L}_{L_\mb{g}h_2}\epsilon(\mb{y})\Vert\mb{u}\Vert_2
\geq-\alpha(h_2(\widehat{\mb{x}})).
\end{align*}
Making use of the fact that $\derp{h_1}{\mb x} = -\derp{h_2}{ \mb x}$, we can simplify and rewrite these constraints as 
\begin{align}
\begin{split}\label{eq:two_cbf_constraints}
L_\mb{g}h_1(\widehat{\mb{x}}) \mb 1 u -  \mathfrak{L}_{L_\mb{g}h}\epsilon(\mb{y})\sqrt{2} |u|
&\geq-\alpha(h_1(\widehat{\mb{x}}))-L_\mb{f}h_1(\widehat{\mb{x}}) +(\mathfrak{L}_{L_\mb{f}h}+\mathfrak{L}_{\alpha\circ h})\epsilon(\mb{y}),\\
 -L_\mb{g}h_1(\widehat{\mb{x}}) \mb 1 u -  \mathfrak{L}_{L_\mb{g}h}\epsilon(\mb{y})\sqrt{2} |u|
&\geq-\alpha(h_2(\widehat{\mb{x}}))+L_\mb{f}h_1(\widehat{\mb{x}}) +(\mathfrak{L}_{L_\mb{f}h}+\mathfrak{L}_{\alpha\circ h})\epsilon(\mb{y}).
\end{split}
\end{align}

\begin{proposition}
The set of constraints~\eqref{eq:two_cbf_constraints} will be feasible for all $\mb x \in \C$ if for all $\mb x \in \C$:
\begin{align*}
      \varepsilon(\mb{x}) <
       \frac{1}{2}\max\Big\{
         &\min\left\{\frac{\alpha(h_1({\mb{x}}))+L_\mb{f}h_1({\mb{x}})}{(\mathfrak{L}_{L_\mb{f}h}+\mathfrak{L}_{\alpha\circ h})},\frac{\alpha(h_2({\mb{x}}))-L_\mb{f}h_1({\mb{x}})}{(\mathfrak{L}_{L_\mb{f}h}+\mathfrak{L}_{\alpha\circ h})}\right\},
         \\
         &\min\Big\{\frac{\alpha(h_1({\mb{x}}))+L_\mb{f}h_1({\mb{x}})}{(\mathfrak{L}_{L_\mb{f}h}+\mathfrak{L}_{\alpha\circ h})},
         \\
         &\,\frac{|L_\mb{g}h_1({\mb{x}}) \mb 1|(\alpha(h_1({\mb{x}}))+ \alpha(h_2({\mb{x}})))}{\mathfrak{L}_{L_\mb{g}h}\sqrt{2}(\alpha(h_2({\mb{x}}))-\alpha(h_1({\mb{x}}))-2L_\mb{f}h_1({\mb{x}}))+2(\mathfrak{L}_{L_\mb{f}h}+\mathfrak{L}_{\alpha\circ h})|L_\mb{g}h_1({\mb{x}}) \mb 1|}\Big\},
         \\
         &
        \min\Big\{\frac{\alpha(h_2({\mb{x}}))-L_\mb{f}h_1({\mb{x}})}{(\mathfrak{L}_{L_\mb{f}h}+\mathfrak{L}_{\alpha\circ h})},
        \\
         &\, \frac{|L_\mb{g}h_1({\mb{x}}) \mb 1|(\alpha(h_1({\mb{x}}))+ \alpha(h_2({\mb{x}})))}{\mathfrak{L}_{L_\mb{g}h}\sqrt{2}(\alpha(h_1({\mb{x}}))-\alpha(h_2({\mb{x}}))+2L_\mb{f}h_1({\mb{x}}))+2(\mathfrak{L}_{L_\mb{f}h}+\mathfrak{L}_{\alpha\circ h})|L_\mb{g}h_1({\mb{x}}) \mb 1|}\Big\}
        \Big\}.
\end{align*}
\end{proposition}
\begin{proof}
First, we adapt constraints in~\eqref{eq:two_cbf_constraints} to be in terms of the true state $\mb x$ rather than the estimated $\widehat{\mb x}$.
By a smoothness argument,
\begin{align*}
L_\mb{g}h_1({\mb{x}} - \mb e) \mb 1 u -  \mathfrak{L}_{L_\mb{g}h}\varepsilon(\mb x)\sqrt{2} |u|
\geq-\alpha(h_1({\mb{x}} - \mb e))-L_\mb{f}h_1({\mb{x}} - \mb e) +(\mathfrak{L}_{L_\mb{f}h}+\mathfrak{L}_{\alpha\circ h})\varepsilon(\mb x) \\
\impliedby L_\mb{g}h_1({\mb{x}}) \mb 1 u - 2 \mathfrak{L}_{L_\mb{g}h}\varepsilon(\mb x)\sqrt{2} |u|
\geq-\alpha(h_1({\mb{x}}))-L_\mb{f}h_1({\mb{x}}) +2(\mathfrak{L}_{L_\mb{f}h}+\mathfrak{L}_{\alpha\circ h})\varepsilon(\mb x)\:.
\end{align*}
A similar argument holds for the second constraint.
Therefore, to find a guarantee depending only on $\mb x$, we effectively double the radius of error to $2\varepsilon(\mb x)$.

Then the result follows from Lemma~\ref{lem:simultaneous_feas}, which applies to this problem with
\begin{align*}
\epsilon &= 2\varepsilon(\mb x),\\
a &= L_\mb{g}h_1({\mb{x}}) \mb 1, \\
b &= \mathfrak{L}_{L_\mb{g}h}\sqrt{2},\\
d_1 &= \alpha(h_1({\mb{x}}))+L_\mb{f}h_1({\mb{x}}),\\
d_2 &= \alpha(h_2({\mb{x}}))-L_\mb{f}h_1({\mb{x}}),\\
L &= (\mathfrak{L}_{L_\mb{f}h}+\mathfrak{L}_{\alpha\circ h}).
\end{align*}
\end{proof}

\begin{lemma}\label{lem:simultaneous_feas}
The following two constraints 
\begin{align*}
a u - b\epsilon |u| \geq -d_1+L\epsilon,\\
-a u - b \epsilon|u| \geq- d_2 + L \epsilon,
\end{align*}
are simultaneously feasible for some $ u \in\R$ if and only if
\begin{align*}
\epsilon\leq \max\left\{
 \min\left\{\frac{d_1}{L},\frac{d_2}{L}\right\},
 \min\left\{\frac{d_1}{L}, \frac{|a|(d_1+d_2)}{b(d_2-d_1)+2L|a|}\right\},
\min\left\{\frac{d_2}{L}, \frac{|a|(d_1+d_2)}{b(d_1-d_2)+2L|a|}\right\}
\right\}\:.
\end{align*}
\end{lemma}
\begin{proof}
It is equivalent to consider 
\begin{align*}
-d_1+L\epsilon \leq \max_u \quad& a u - b\epsilon |u| \\
\text{s.t.}\quad &-a u - b\epsilon |u| \geq -d_2+L\epsilon.
\end{align*}
First, notice that the first condition implies that  $u=0$ satisfies the expression.

Otherwise, consider the case that $-d_2+L\epsilon\leq 0$, and see that we will set $u$ to have the same sign as $a$.
Then the optimization problem is solved by
\begin{align*}
(|a|  - b\epsilon)\frac{-d_2+L\epsilon}{ - (|a| + b\epsilon)} = \max_{|u|} \quad& |u|(|a|  - b\epsilon) \\
\text{s.t.}\quad &- (|a| + b\epsilon) |u| \geq -d_2+L\epsilon.
\end{align*}
This yields the third condition after some algebra:
\begin{align*}
(|a|  - b\epsilon)\frac{-d_2+L\epsilon}{ - (|a| + b\epsilon)} &\geq -d_1+L\epsilon,\\
(|a|  - b\epsilon) (d_2-L\epsilon)&\geq (-d_1+L\epsilon)(|a| + b\epsilon),\\
|a|d_2  - b\epsilon d_2 -L\epsilon|a| +bL\epsilon^2 &\geq -d_1|a|+L\epsilon|a| + -d_1b\epsilon+L b\epsilon^2,\\
|a|(d_2+d_1)    &\geq (2 L|a| -d_1b + d_2 b)\epsilon.
\end{align*}
Reversing the roles of $d_1$ and $d_2$ and setting $u$ to have the opposite sign of $a$ yields the second condition.
\end{proof}

 \end{document}